\documentclass[11pt,reqno]{amsart}
\usepackage[left=3cm,top=2cm,right=3cm,bottom=2cm]{geometry}
\usepackage{amsbsy}
\usepackage{tikz}
\usepackage{tkz-berge}
\usepackage[mathscr]{euscript}
\usepackage{amssymb,amsthm}
\usepackage{multirow}
\usepackage{relsize}
\usepackage{bm}
\usepackage{arydshln}
\usepackage{setspace}

\usepackage[mathlines]{lineno}
\usepackage{etoolbox} 

\newcommand*\linenomathpatch[1]{%
   \expandafter\pretocmd\csname #1\endcsname {\linenomath}{}{}%
   \expandafter\pretocmd\csname #1*\endcsname{\linenomath}{}{}%
   \expandafter\apptocmd\csname end#1\endcsname {\endlinenomath}{}{}%
   \expandafter\apptocmd\csname end#1*\endcsname{\endlinenomath}{}{}%
 }
\newcommand*\linenomathpatchAMS[1]{%
    \expandafter\pretocmd\csname #1\endcsname {\linenomathAMS}{}{}%
    \expandafter\pretocmd\csname #1*\endcsname{\linenomathAMS}{}{}%
    \expandafter\apptocmd\csname end#1\endcsname {\endlinenomath}{}{}%
    \expandafter\apptocmd\csname end#1*\endcsname{\endlinenomath}{}{}%
}

\expandafter\ifx\linenomath\linenomathWithnumbers
\let\linenomathAMS\linenomathWithnumbers
\patchcmd\linenomathAMS{\advance\postdisplaypenalty\linenopenalty}{}{}{}
\else
\let\linenomathAMS\linenomathNonumbers
\fi

\linenomathpatchAMS{gather}
\linenomathpatchAMS{multline}
\linenomathpatchAMS{align}
\linenomathpatchAMS{alignat}
\linenomathpatchAMS{flalign}
\linenomathpatch{equation}

\AtBeginDocument{%
   \def\MR#1{}
}

\newcommand{\typei}{type-i}
\newcommand{\typeii}{type-ii}

\DeclareMathOperator*{\rank}{rank}
\DeclareMathOperator*{\inertia}{inertia}

\DeclareMathOperator{\twop}{tw}

\newtheorem{theorem}{Theorem}
\newtheorem{lemma}{Lemma}

\newtheorem{claim}{Claim}

\usepackage{lineno}
\newcommand*\patchAmsMathEnvironmentForLineno[1]{%
\expandafter\let\csname old#1\expandafter\endcsname\csname #1\endcsname
\expandafter\let\csname oldend#1\expandafter\endcsname\csname end#1\endcsname
\renewenvironment{#1}%
{\linenomath\csname old#1\endcsname}%
{\csname oldend#1\endcsname\endlinenomath}}%
\newcommand*\patchBothAmsMathEnvironmentsForLineno[1]{%
\patchAmsMathEnvironmentForLineno{#1}%
\patchAmsMathEnvironmentForLineno{#1*}}%
\AtBeginDocument{%
\patchBothAmsMathEnvironmentsForLineno{equation}%
\patchBothAmsMathEnvironmentsForLineno{align}%
\patchBothAmsMathEnvironmentsForLineno{flalign}%
\patchBothAmsMathEnvironmentsForLineno{alignat}%
\patchBothAmsMathEnvironmentsForLineno{gather}%
\patchBothAmsMathEnvironmentsForLineno{multline}%
}

\begin{document}
\onehalfspace

\title[Efficient diagonalization of symmetric matrices]{Efficient diagonalization of symmetric matrices associated with graphs of small treewidth}
\author{Martin F\"{u}rer}
\address{Department of Computer Science and Engineering, Pennsylvania State University}
\email{\tt fhs@psu.edu}
\author{Carlos Hoppen}
\address{Instituto de Matem\'atica e Estat\'{i}stica, Universidade Federal do Rio Grande do Sul}
\email{\tt choppen@ufrgs.br}
\author{Vilmar Trevisan}
\address{Instituto de Matem\'atica e Estat\'{i}stica, Universidade Federal do Rio Grande do Sul}
\email{\tt trevisan@mat.ufrgs.br}
\pdfpagewidth 8.5 in
\pdfpageheight 11 in

\begin{abstract}
Let $M=(m_{ij})$ be a symmetric matrix of order $n$ whose elements lie in an arbitrary field $\mathbb{F}$, and let $G$ be the graph with vertex set $\{1,\ldots,n\}$ such that distinct vertices $i$ and $j$ are adjacent if and only if $m_{ij} \neq 0$.  We introduce a dynamic programming algorithm that finds a diagonal matrix that is congruent to $M$. If $G$ is given with a tree decomposition $\mathcal{T}$ of width $k$, then this can be done in time $O(k|\mathcal{T}| + k^2 n)$, where $|\mathcal{T}|$ denotes the number of nodes in $\mathcal{T}$. Among other things, this allows the computation of the determinant, the rank and the inertia of a symmetric matrix in time $O(k|\mathcal{T}| + k^2 n)$.
\end{abstract}

\thanks{A conference version of this paper appeared in the Proceedings of the $47^{th}$ International Colloquium on Automata, Languages and Programming (ICALP 2020)~\cite{icalp}.}

\thanks{C. Hoppen acknowledges the support of CNPq~308054/2018-0 and FAPERGS~19/2551-0001727-8. V. Trevisan acknowledges partial support of CNPq grants 409746/2016-9 and 303334/2016-9, CAPES under project MATHAMSUD 18-MATH-01 and FAPERGS under Project  PqG 17/2551-0001. CNPq is Conselho Nacional de Desenvolvimento Cient\'{i}fico e Tecnol\'{o}gico, CAPES is Coordena\c{c}\~{a}o de Aperfei\c{c}oamento de Pessoal de N\'{i}vel Superior and FAPERGS is Funda\c{c}\~{a}o de Amparo \`{a} Pesquisa do Estado do Rio Grande do Sul.}

\maketitle

\section{Introduction}
\label{sec:intro}
The aim of this paper is to associate symmetric matrices with diagonal matrices by performing elementary operations. Throughout the paper, we shall assume that readers are familiar with elementary row and column operations, matrices in row echelon form and Gaussian elimination, as defined in~\cite{Meyer2000,Stewart1995} or in most textbooks of basic linear algebra. Suppose that we start with a square matrix $M$ whose elements lie in an arbitrary field $\mathbb{F}$ and sequentially perform elementary operations in rounds, so that at each round we perform an elementary row operation followed by the same column operation. This produces a matrix $N = P M P^T$, where $P$ is a non-singular matrix. We say that $M$ and $N$ are {\em congruent}, which is denoted by $M \cong N$. Given a symmetric matrix $M$, we wish to \emph{diagonalize} $M$, by which we mean to find a diagonal matrix $D$ that is congruent to $M$.

Finding a diagonal matrix $D$ that is congruent to a symmetric matrix $M$ of order $n$ using elementary row and column operations is a basic operation in linear algebra. It allows the computation of the determinant and of the rank of $M$, for instance. It may also be used to compute the \emph{inertia} of a matrix, namely the triple $(n_+,n_{-},n_0)$ where $n_+$, $n_{-}$ and $n_0$ are the number of positive, negative and zero eigenvalues of $M$, respectively, with multiplicity. Moreover, matrix congruence naturally appears in the context of Gram matrices associated with a quadratic form (or bilinear form) on a finite-dimensional vector space. In this context, two matrices are congruent if and only if they represent the same quadratic form with respect to different bases, and finding a diagonal matrix $D$ that is congruent to a symmetric matrix $M$ allows the classification of the quadratic form.

Any symmetric matrix $M=(m_{ij})$ of order $n$ may be naturally associated with a graph $G$ with vertex set $[n]=\{1,\ldots,n\}$ such that distinct vertices $i$ and $j$ are adjacent if and only if $m_{ij} \neq 0$. We say that $G$ is the \emph{underlying graph} of $M$. This allows us to employ structural decompositions of graph theory to deal with the nonzero entries of $M$ in an efficient way. One such decomposition is the tree decomposition, which has been extensively studied since the seminal paper of Robertson and Seymour~\cite{RobertsonS86}, but had already been independently studied in various contexts since the seventies, see \cite{BB1972,Halin1976}. The graph parameter associated with this decomposition is called \emph{treewidth}.

Having a tree decomposition of a graph $G=(V,E)$ with small width has proved to be quite useful algorithmically. It has been used to design algorithms for NP-complete or even harder problems that are efficient on graphs of bounded width, by which we mean that their running time is given by $O(f(k) n^c)$ for a constant $c$ and an arbitrary computable function $f$, where $k$ is the width of the graph. In complexity theory, this means that the problems are fixed parameter tractable (FPT). Even if $f$ is at least exponential, such algorithms are often quite practical for graphs with small treewidth. We refer to~\cite{CyganFKLMPPS2015,DowneyF2013,FlumG2006,Niedermeier2006} for a general introduction to fixed parameter tractability and to~Bodlaender and Koster~\cite{Bodlaender:2008} for FPT algorithms based particularly on the treewidth. Unfortunately, computing the treewidth of a graph $G$ is NP-hard in general~\cite{ACP87}; however, Bodlaender~\cite{Bodlaender1996} has shown that, for any fixed constant $k$, there is an algorithm with running time $f(k)n$ which decides whether an $n$-vertex graph has treewidth at most $k$, and outputs a corresponding tree decomposition if the answer is positive (where $f(k)=2^{O(k^3)}$). More recently, Fomin et al.\ \cite{FominLSPW2018} described an algorithm that has running time $O(k^7 n\log{n})$ and either correctly reports that the treewidth of $G$ is larger than $k$, or constructs a tree decomposition of $G$ of width $O(k^2)$. When considering the complexity of an algorithm that uses a tree decomposition of the input graph, we assume that the tree decomposition is given along with the input graph.

One of the early applications of tree decompositions has been in Gaussian elimination, particularly when the input matrix is sparse. To take advantage of the sparsity, the goal is to find a pivoting scheme (or elimination ordering) that minimizes the \emph{fill-in}, defined as the set of matrix positions that were initially 0, but became nonzero at some point during the computation. For instance, to minimize the fill-in during Gaussian elimination, Rose~\cite{Rose1972}, and Rose, Tarjan and Lueker~\cite{RoseTL1976} considered symmetric matrices (see also~\cite{RoseT1978} for matrices that are not necessarily symmetric) and associated them with an underlying graph $G$ as in the present paper. Among other things, they devised a linear-time algorithm to find a \emph{total elimination ordering} in $G$\footnote{Their algorithm works for all graphs for which such an ordering exists, which are precisely the chordal graphs~\cite{FG65}.}, namely an ordering of the vertices such that, for any vertex $v \in V(G)$, the set of neighbors of $v$ that appear after $v$ in the ordering forms a clique. It turns out that no fill-in is produced when performing Gaussian elimination using rows in this order, assuming that the diagonal of the input matrix is nonzero and that no accidental cancellations occur (an \emph{accidental cancellation} happens when a nonzero element becomes zero as a by-product of an operation aimed at cancelling a different element on its row). We would like to point out that, despite seeming to be helpful at first sight, accidental cancellations may mess up with the pivoting scheme, as an unexpected zero in the diagonal can force the algorithm to use a different pivot than originally planned, potentially leading to fill-in.

Based on these ideas, Radhakrishnan, Hunt and Stearns~\cite{RHS92} designed an $O(k^2n)$ algorithm to turn a matrix $M$ of order $n$ whose underlying graph has treewidth at most $k$ into row echelon form, under the assumptions that no accidental cancellations occur and that a tree decomposition of width at most $k$ is part of the input. We note that a worst-case bound of $O(k^2n)$ is best possible for algorithms based on elementary row operations. If the assumption about accidental cancellations is dropped, Fomin et al.~\cite{FominLSPW2018}\footnote{The matrices in~\cite{FominLSPW2018} are not necessarily symmetric (in fact, not necessarily square), and they associate an $m \times n$-matrix with a bipartite graph whose partition classes have size $m$ and $n$.} devised a clever pivoting scheme to get row echelon form with relatively small fill-in, but their algorithm requires $O(k^3 n)$ field operations in the worst case. On the other hand, it achieves the best possible bound of $O(k^2n)$ for graphs with a path or a tree-partition decomposition of width $k$.

We note that both ~\cite{RHS92} and~\cite{FominLSPW2018} fit into the recent trend of ``FPT within P'', which investigates fundamental problems that are solvable in polynomial time, but for which a lower exponent may be achieved using an FPT algorithm that is also polynomial in terms of the parameter $k$ (see~\cite{GiannopoulouMN2017}). The same is true for our paper.

Our aim is not to perform Gaussian elimination, but instead to find a diagonal matrix that is congruent to an input symmetric matrix $M$ of order $n$ whose underlying graph has treewidth $k$. For this problem, we do not make any assumption about the absence of accidental cancellations and our algorithm runs in time $O(k^2n)$ provided that the input includes a ``compact'' tree decomposition of the underlying graph. This may be formalized as follows.
\begin{theorem}\label{main_theorem}
Given a symmetric matrix $M$ of order $n$ and a tree decomposition $\mathcal{T}$ of width $k$ for
the underlying graph of $M$, algorithm \texttt{CongruentDiagonal} (see Sect.~\ref{sec:algorithm}) produces a diagonal matrix $D$ congruent to $M$ in time $O(k|\mathcal{T}| + k^2n)$.
\end{theorem}
As it turns out, a pre-processing step of the algorithm turns the given tree decomposition into a \emph{nice tree decomposition}, in the sense of Kloks~\cite{Kloks1994}. Roughly speaking, the algorithm then tries to follow a pivoting scheme induced by the tree decomposition. To deal with accidental cancellations, we do not change the scheme, but rather create a temporary buffer for rows and columns for which the diagonal entry that was supposed to be a pivot happens to be zero. We ensure that the temporary buffer remains small throughout the algorithm by performing some additional operations. It should also be mentioned that the fact that we are allowed to perform both row and column operations in a symmetric matrix is crucial for our approach, and it is not clear how to handle matrices that are not symmetric or how to adapt this approach to Gaussian elimination.

To conclude the introduction, we mention an application of our work to spectral graph theory. This is a branch of graph theory that aims to extract structural information about graphs from algebraic information about different types of matrices associated with them, such as the adjacency, the distance and the Laplacian matrices, to name a few. These examples, like many others, are real-valued symmetric matrices. When $\mathbb{F}$ is a subfield of the set $\mathbb{R}$ of real numbers, determining a diagonal matrix that is congruent to a symmetric matrix $M$ is also quite useful for determining the number of eigenvalues of $M$ in any given real interval, as a consequence of Sylvester's Law of Inertia~\cite[p. 568]{Meyer2000}. In this respect, the current paper may be viewed as a contribution in the line of research of \emph{eigenvalue location} for matrices associated with graphs. We say that an algorithm \emph{locates eigenvalues for a class $\mathcal{C}$} (with respect to a type of matrix $M(G)$ associated with graphs $G$ in a certain class) if, for any graph $G \in \mathcal{C}$ and any given real interval $I$, it finds the number of eigenvalues of $M(G)$ in the interval $I$. Regarding the adjacency matrix, in recent years, linear time algorithms have been developed for eigenvalue location in trees~\cite{JT2011}, threshold graphs~\cite{JTT2013}, chain graphs~\cite{MR3506491} and cographs~\cite{JacobsTT2015}, among others. An application of the algorithm in this paper is a linear time eigenvalue location algorithm for any symmetric matrix $M=(m_{ij})$ whose underlying graph has bounded treewidth. Recently, Jacobs and the current authors~\cite{FHJT2019} obtained an eigenvalue location algorithm for certain matrices associated with graphs of bounded clique-width (the \emph{clique-width} is defined based on a structural decomposition introduced by Courcelle and Olariu \cite{CourcelleO2000}).  In addition to dealing with a different decomposition, the results in the current paper are much stronger, as they hold for any symmetric matrix, while the algorithm ~\cite{FHJT2019} applies only to symmetric matrices such that all nonzero off-diagonal entries have the same value. In the restricted setting of spectral graph theory, this means that the results in~\cite{FHJT2019} apply to the adjacency, Laplacian and signless Laplacian matrices associated with a graph, for instance, but not to the normalized Laplacian or to distance-based matrices, while our algorithm applies to all of them. In fact, an $n \times n$ symmetric matrix $M=(m_{ij})$  may be viewed as a weighted graph with vertex set $[n]$ where each vertex $i$ has weight $m_{ii}$ and each pair $\{i,j\}$, where $i \neq j$, is an edge if and only if $m_{ij} \neq 0$, in which case $m_{ij}$ is the weight of the edge. On this understanding, our paper applies to weighted graphs.

The paper is organized as follows. In Section~\ref{sec:preliminaries}, we define the tree decompositions that will be used in this paper and describe the properties that will be useful for our purposes. The algorithm is described in Section~\ref{sec:algorithm}, and we illustrate it with an example in Section~\ref{sec:example}. To conclude the paper, we prove the correctness of the algorithm in Section~\ref{sec:proof} and analyse its running time in Section~\ref{sec:analysis}. Together, these two sections establish Theorem~\ref{main_theorem}.

\section{Tree decompositions}\label{sec:preliminaries}

Let $G=(V,E)$ be an $n$-vertex graph, with the standard assumption that $V=[n]$. A \emph{tree decomposition} of a graph $G$ is a tree $\mathcal{T}$ with \emph{nodes} $1,\ldots,m$, where each node $i$ is associated with a \emph{bag} $B_i \subseteq V$, satisfying the following properties:
\begin{enumerate}
\item $\bigcup_{i=1}^m B_i=V$.
\item For every edge $\{v,w\} \in E$, there exists $B_i$ containing $v$ and $w$.
\item For any $v \in V$, the subgraph of $\mathcal{T}$ induced by the nodes that contain $v$ is connected.
\end{enumerate}
The \emph{width} of the tree decomposition $\mathcal{T}$ is defined as $\max_i \left(|B_i|-1\right)$ and the \emph{treewidth} $\twop(G)$ of graph $G$ is the smallest $k$ such that $G$ has a tree decomposition of width $k$. Clearly, it is always the case that $\twop(G) \leq |V|-1$; moreover, for connected graphs, $\twop(G)=1$ if and only if $G$ is a tree on two or more vertices. It is well known that having small treewidth implies the sparsity of the graph. Indeed, any $n$-vertex graph $G=(V,E)$ such that $\twop(G) \leq k$ satisfies $|E|\leq kn-\binom{k+1}{2}$.

In this paper, we actually use the concept of \emph{nice tree decomposition}, introduced by Kloks~\cite{Kloks1994}, which is a rooted tree decomposition $\mathcal{T}$ of a graph $G$ such that all nodes are of one of the following types:
\begin{itemize}
\item[(a)] \textbf{(Leaf)} The node $i$ is a leaf of $\mathcal{T}$;

\item[(b)] \textbf{(Introduce)} The node $i$ introduces vertex $v$, that is, it has a single child $j$, $v \notin B_j$ and $B_i=B_j \cup \{v\}$.

\item[(c)] \textbf{(Forget)} The node $i$ forgets vertex $v$, that is, $i$ has a single child $j$, $v \notin B_i$ and $B_j=B_i \cup \{v\}$;

\item[(d)] \textbf{(Join)} The node $i$ is a join, that is, it has two children $B_j$ and $B_\ell$, where $B_i=B_j=B_\ell$.
\end{itemize}
For convenience, we further assume that the root has an empty bag. This is easy to accomplish, as any tree decomposition $\mathcal{T}$ that satisfies (a)-(d), but does not satisfy this additional requirement, may be turned into a nice tree decomposition with empty root by appending a path of length at most $k+1$ to the root of $\mathcal{T}$, where all nodes have type Forget. Having a root with an empty bag is convenient to ensure that each vertex of $G=(V,E)$ is associated with exactly one Forget node that forgets it. Moreover, it ensures that, for any $v \in V$, the node of $\mathcal{T}$ that is closest to the root among all $j$ such that $v \in B_j$ is the child of the node that forgets $v$. Figure~\ref{figure:ntc} depicts a graph and a nice tree decomposition associated with it.

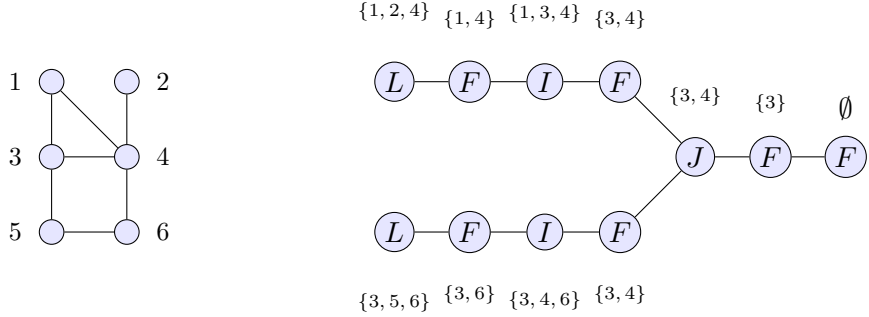
\begin{figure}
\begin{minipage}[l]{0.3 \linewidth}
\begin{tikzpicture}
[scale=1,auto=left,every node/.style={circle,scale=0.9}]
  \node[draw,circle,fill=blue!10,label=right:$2$] (a) at (0,1) {};
   \node[draw,fill=blue!10,circle,label=left:$1$] (b) at (-1,1) {};
 \node[draw,fill=blue!10,circle,label=right:$4$] (d) at (0,0) {};
 \node[draw,fill=blue!10, circle,label=left:$3$] (c) at (-1,0) {};
 \node[draw,fill=blue!10,circle,label=left:$5$] (e) at (-1,-1) {};
  \node[draw,fill=blue!10, circle,label=right:$6$] (f) at (0,-1) {};
  \path
          (a) edge node[above]{} (d)
          (c) edge node[above]{} (b)
          (c) edge node[above]{} (d)
          (c) edge node[above]{} (e)
          (d) edge node[above]{} (f)
          (d) edge node[above]{} (b)
          (e) edge node[above]{} (f);
\end{tikzpicture}
       \end{minipage}\hfill
       \begin{minipage}[l]{0.7 \linewidth}
\begin{tikzpicture}
  [scale=1,auto=left,every node/.style={circle,scale=1}]
  \node[draw,circle,fill=blue!10, inner sep=1.5,label=above:{\tiny{$\{1,2,4\}$}}] (a) at (-3,1) {$L$};
  \node[draw,circle,fill=blue!10, inner sep=1.5,label=above:{\tiny{$\{1,4\}$}}] (b) at (-2,1) {$F$};
  \node[draw,circle,fill=blue!10, inner sep=1.5,label=above:{\tiny{$\{1,3,4\}$}}] (c) at (-1,1) {$I$};
    \node[draw,circle,fill=blue!10, inner sep=1.5,label=below:{\tiny{$\{3,5,6\}$}}] (d) at (-3,-1) {$L$};
  \node[draw,circle,fill=blue!10, inner sep=1.5,label=below:{\tiny{$\{3,6\}$}}] (e) at (-2,-1) {$F$};
  \node[draw,circle,fill=blue!10, inner sep=1.5,label=below:{\tiny{$\{3,4,6\}$}}] (f) at (-1,-1) {$I$};
  \node[draw,circle,fill=blue!10, inner sep=1.5,label=above:{\tiny{$\{3,4\}$}}] (k) at (0,1) {$F$};
   \node[draw,circle,fill=blue!10, inner sep=1.5,label=below:{\tiny{$\{3,4\}$}}] (l) at (0,-1) {$F$};
   \node[draw,circle,fill=blue!10, inner sep=1.5,label=above:{\tiny{$\{3,4\}$}}] (g) at (1,0) {$J$};
   \node[draw,circle,fill=blue!10, inner sep=1.5,label=above:{\tiny{$\{3\}$}}] (h) at (2,0) {$F$};
     \node[draw,circle,fill=blue!10, inner sep=1.5,label=above:{$\emptyset$}] (i) at (3,0) {$F$};
  \path (a) edge node[below]{} (b)
        (b) edge node[below]{} (c)
        (d) edge node[below]{} (e)
        (e) edge node[below]{} (f)
        (c) edge node[below]{} (k)
        (f) edge node[below]{} (l)
        (h) edge node[below]{} (g)
        (h) edge node[below]{} (i)
        (g) edge node[below]{} (k)
        (g) edge node[below]{} (l);
\end{tikzpicture}
\end{minipage}
\caption{A graph and a nice tree decomposition associated with it. The root of the tree is on the right and the nodes are of type $L$ (leaf), $I$ (introduce), $F$ (forget) or join (J). We view the top child of the join node as its left child, and the bottom child as its right child.}
\label{figure:ntc}
\end{figure}

Despite the additional structure, a nice tree decomposition may be efficiently derived from an arbitrary tree decomposition. More precisely, Kloks~\cite[Lemma~13.1.2]{Kloks1994} has shown that, if $G$ is a graph of order $n$ and we are given an arbitrary tree decomposition of $G$ with width $k$ and $m$ nodes, it is possible to turn it into a nice tree decomposition of $G$ with at most $4n$ nodes and width at most $k$ in time $O(k(m+n))$. This also holds when the root is required to have an empty bag. For completeness, we state a more precise version as a lemma. A proof may be found in the appendix.
\begin{lemma}\label{lemma_nice:tree}
Let $G$ be a graph of order $n$ with a tree decomposition $\mathcal{T}$ of width $k$ with $m$ nodes. Then $G$ has a nice tree decomposition of the same width $k$ with fewer than $4n-1$ nodes that can be computed from $\mathcal{T}$ in time $O(k(m+n))$.
\end{lemma}

\section{The Algorithm}
\label{sec:algorithm}
We now describe our diagonalization algorithm, which we call \texttt{CongruentDiagonal}. Let $M$ be a symmetric matrix of order $n$ and let $G = (V,E)$ be the underlying graph with vertex set $V=[n]$ associated with it.  We wish to find a diagonal matrix congruent to $M$ in linear time, so that we shall operate on a sparse representation of the graph associated with $M$, rather than on the matrix itself. Let $\mathcal{T}$ be a nice tree decomposition of $G$ with node set $[m]$ and width $k$. The algorithm \texttt{CongruentDiagonal} works bottom-up in the rooted tree $\mathcal{T}$, that is, it only processes a node $i$ after its children have been processed. Each node $i$ other than the root produces a data structure known as a \emph{box} and transmits it to its parent. A box is a symmetric matrix $N_i$ of order at most $2(k+1)$ that may be recorded as pair of matrices $(N_i^{(1)},N_i^{(2)})$ whose rows and columns are labeled by elements of $[n]$. These labels relate rows and columns of the boxes with rows and columns of the original matrix. While producing the box, the algorithm may also produce diagonal elements of a matrix congruent to $M$. These diagonal elements are not transmitted to the node's parent, but are appended to a global array. At the end of the algorithm, the array consists of the diagonal elements of a diagonal matrix $D$ that is congruent to $M$. Figure~\ref{fig:highlevel} shows a high-level description of the algorithm.

To get an intuition about how the algorithm works, assume that we are trying to perform Gaussian elimination, but performing the \emph{same} column operation after each row operation to produce a matrix that is congruent to the original matrix. The pivoting scheme that we choose is based on the tree decomposition $\mathcal{T}$ as follows. Recall that each row of $M$ is associated with a vertex $v$ of $G$, which in turn is associated with the unique Forget node of $\mathcal{T}$ that forgets $v$. Given a bottom-up ordering of $\mathcal{T}$, the first row in the pivoting scheme is the row associated with the first Forget node in this ordering, the second row is the one associated with the second Forget node, and so on. The algorithm attempts to diagonalize row $v$ at the node where $v$ is forgotten. If the diagonal entry $d_v$ corresponding to $v$ in the box is nonzero, the algorithm uses it to annihilate nonzero elements in this row/column, and $(v,d_v)$ is appended to the global array mentioned above. It turns out that, if the same row and columns operations had been performed in the entire matrix, this would lead to the row and column associated with $v$ being diagonalized with diagonal entry $d_v$. However, if $d_v=0$, the original entry $m_{vv}$ was 0 or an accidental cancellation has occurred. Then the row and column associated with $v$ will be added to a temporary buffer, and it will later be diagonalized in the process of producing new boxes.

\begin{figure}[h]
{\tt
\begin{tabbing}
aaa\=aaa\=aaa\=aaa\=aaa\=aaa\=aaa\=aaa\= \kill
     \> \texttt{CongruentDiagonal}(M)\\
     \> {\bf input}:  a nice tree decomposition $\mathcal{T}$ of the underlying  \\
     \> graph $G$ associated with $M$ of width $k$ and the nonzero entries of $M$ \\
     \> {\bf output}: diagonal entries in $D \cong M$ \\
     \> Order the nodes of $\mathcal{T}$ as $1,2,\ldots,m$ in post order\\
     \>{\bf for} $i$ {\bf from } 1 \bf {to} $m$ {\bf do} \\
     \> \> {\bf if} is-Leaf($B_i$) {\bf then} $(N_i^{(1)},N_i^{(2)})$=\texttt{LeafBox}($B_i$)\\
     \> \> {\bf if} is-Introduce($B_i$)  {\bf then} $(N_i^{(1)},N_i^{(2)})$=\texttt{IntroBox}($B_i$) \\
     \> \> {\bf if} is-Join($B_i$)  {\bf then} $(N_i^{(1)},N_i^{(2)})$=\texttt{JoinBox}($B_i$) \\
    \> \> {\bf if} is-Forget($B_i$)  {\bf then} $(N_i^{(1)},N_i^{(2)})$=\texttt{ForgetBox}($B_i$) \\
\end{tabbing}}
\caption{High level description of the algorithm \texttt{CongruentDiagonal}.}\label{fig:highlevel}
\end{figure}

In the remainder of this section, we shall describe each operation in detail. To help us do this, we first give a more detailed description of the boxes produced at each node of the tree. To this end, we first say what we mean by a (not necessarily symmetric) matrix  $A=(a_{ij})$ being in \emph{row echelon form}. This means that
\begin{itemize}
\item[(i)] $a_{ij}=0$ for all $j<i$;

\item[(ii)] The \emph{pivot} of row $i$ is the element $a_{ij} \neq 0$ with least $j$, if it exists. If rows $i_1<i_2$ have pivots in columns $j_1$ and $j_2$, respectively, then  $j_1 <j_2$.
\end{itemize}
Each box $N_i$ produced by a node $i$ has the form
\begin{equation}
\label{eq:formofM}
N_i = \begin{array}{|c|cc|}\hline
& & \\
N_i^{(0)} & \hspace*{3mm} & N_i^{(1)} \hspace*{5mm}\\
& & \\\hline

 & & \\

N_i^{(1)T} & \hspace*{3mm} & N_i^{(2)} \hspace*{5mm}\\
& & \\\hline
\end{array},
\end{equation}
where $N_i^{(0)}$ is a matrix of dimension $k_i' \times k_i'$ whose entries are zero, $N_i^{(2)}$ is a symmetric matrix of dimension $k_i'' \times k''_i$ and $N_i^{(1)}$ is a $k_i' \times k_i''$ matrix in row echelon form such that every row has a pivot. Moreover, $0 \leq k_i' \leq k_i''=|B_i| \leq k+1$. We allow $k'_i$ or  $k_i''$ to be zero, in which case we regard $N_i^{(0)}$, $N_i^{(1)}$ or $N_i^{(2)}$ as empty. As mentioned above, each row of $N_i$ (and the corresponding column) is associated with a vertex of $G$ (equivalently, a row of $M$). Let $V(N_i)$ denote the set of vertices of $G$ associated with the rows of $N_i$. Following the terminology in~\cite{FHJT2018}, we say that the $k_i'$ rows in $N_i^{(0)}$ (or $N_i^{(1)}$) have {\em \typei} and the $k_i''$ rows of $N_i^{(2)}$ have {\em \typeii}. This is represented by the partition $V(N_i)=V_1(N_i) \cup V_2(N_i)$, where $V_1(N_i)$ and $V_2(N_i)$ are the vertices of \typei{} and \typeii{}, respectively. As it turns out, the vertices of \typeii{} are precisely the vertices in $B_i$, while the vertices in $V_1(N_i)$ correspond to the rows that have been added to a ``temporary buffer'' due to an accidental cancellation and have yet to be diagonalized.

For future reference, the structure of the box $N_i$ is described in the following lemma. It will be established as part of the proof of correctness of the algorithm.
\begin{lemma}\label{lemma_Ni}
For all $i \in [m]$, the matrix $N_i$ defined in terms of the pair $(N_i^{(1)},N_i^{(2)})$ produced by node $B_i$ satisfies the following properties:
\begin{itemize}
\item[(a)] $0 \leq k_i' \leq k_i'' \leq k+1$.

\item[(b)]  $N_i^{(1)}$ is a matrix in row echelon form such that every row has a pivot.

\item[(c)] $N_i^{(2)}$ is a symmetric matrix such that $V_2(N_i)=B_{i}$ and whose rows and columns are sorted in increasing order.
\end{itemize}
\end{lemma}

We next describe each step of algorithm \texttt{CongruentDiagonal} in detail.

\vspace{10pt}

\noindent \textbf{LeafBox.} When the node is a leaf corresponding to a bag $B_i$ of size $b_i$, then we apply procedure \texttt{LeafBox} of Figure~\ref{fig:leafbox}. This procedure initializes a box $N_i$ where $k'=0$ and $k''=b_i$, where $N_i^{(2)}$ is the zero matrix of dimension $k'' \times k''$ whose rows and columns are labeled by the elements of $B_i$.
\begin{figure}[h]
{\tt
\begin{tabbing}
aaa\=aaa\=aaa\=aaa\=aaa\=aaa\=aaa\=aaa\= \kill
     \> \texttt{LeafBox}($B_i$)\\
     \> {\bf input}: bag $B_i$ of size $b_i$ \\
     \> {\bf output}: a matrix $N_i=(N_i^{(1)},N_i^{(2)})$ \\
     \> \> Set $N_i^{(1)} =\emptyset$ \\
     \> \> $N_i^{(2)}$ is a zero matrix of dimension $b_i \times b_i$.\\
\end{tabbing}}
\caption{Procedure \texttt{LeafBox}.}\label{fig:leafbox}
\end{figure}

One might be tempted to put $m_{uv}$ at position $u,v$ of $N_i^{(2)}$. However, it is better to add $m_{uv}$ in the step associated with forgetting $u$ or $v$ in order to keep the Join operation simple.

\vspace{10pt}

\noindent \textbf{IntroduceBox.}  When the node $i$ introduces vertex $v$, its bag satisfies $B_i=B_j \cup \{v\}$, where
$B_j$ (which does not contain $v$) is the bag of its child $j$. The input of  \texttt{IntroduceBox} is the vertex $v$ that has been introduced and the box $N_j$ induced by the pair $(N_j^{(1)},N_j^{(2)})$ transmitted by its child, and the box $N_i$ simply adds a new type-ii row/column labeled by $v$ whose elements are all zero, taking care to insert it in the correct order with respect to the labels of type-ii vertices, that is, the labels of the rows of $N_i^{(2)}$ are in increasing order. This is described in Figure~\ref{fig:introbox}.
\begin{figure}[h]\label{fig:introbox}
{\tt
\begin{tabbing}
aaa\=aaa\=aaa\=aaa\=aaa\=aaa\=aaa\=aaa\= \kill
     \> \texttt{IntroduceBox}($B_i,v,N_j$)\\
     \> {\bf input}:  bag $B_i$, vertex $v$ and bag $N_j=(N_j^{(1)},N_j^{(2)})$ produced by its child \\
     \> {\bf output}: a matrix $N_i=(N_i^{(1)},N_i^{(2)})$ \\
     \> \> $N_i^{(1)}=N_j^{(1)}$, $N_i^{(2)} =N_j^{(2)}$\\
     \> \> add zero row and zero column associated with $v$ to $N_i^{(2)}$\\
     \>\>\> making sure that type-ii rows remain in increasing order\\
     \> \> if $N_i^{(1)}$ is nonempty, add zero column to $N_i^{(1)}$\\
     \>\>\> preserving the column order of $N_i^{(2)}$\\
\end{tabbing}}
\caption{Procedure \texttt{IntroduceBox}.}\label{fig:introbox}
\end{figure}

\vspace{10pt}

\noindent \textbf{JoinBox.} Let $i$ be a node of type join and let $N_j$ and $N_\ell$ be the boxes induced by the pairs transmitted by its children, where $j<\ell<i$. By Lemma~\ref{lemma_Ni}(c) and the definition of the join operation, we have $V_2(N_j) = V_2(N_\ell)$. Moreover, when proving the correctness of the algorithm, we will establish that $V_1(N_j) \cap V_1(N_\ell)=\emptyset$. The \texttt{JoinBox} operation first creates a matrix $N_i^\ast$ whose rows and columns are labeled by $V_1(N_j) \cup V_1(N_\ell) \cup V_2(N_j)$ with the structure below. Assume that $|V_1(N_j)|=r$, $|V_1(N_\ell)|=s$ and $|B_i|=t$, and define
\begin{equation}
\label{eq:formofNast}
N_i^\ast = \begin{array}{|c|c|c|}
\hline
\mathbf{0}_{r \times r} & \mathbf{0}_{r \times s} & N_j^{(1)} \\ \hline
 \mathbf{0}_{s \times r} &  \mathbf{0}_{s \times s}& N_\ell^{(1)} \\ \hline
 N_j^{(1)T} & N_\ell^{(1)T} & N_i^{\ast(2)}\\ \hline
\end{array},
\end{equation}
where $N_i^{\ast(2)}=N_j^{(2)}+N_\ell^{(2)}$. The matrices $N_j^{(1)}$ and $N_\ell^{(2)}$ simply appear on top of each other because we made sure that its columns have the same labeling, which is why we require the labels of type-ii rows to be in increasing order. Let $N_i^{\ast(0)}$ denote the zero matrix of dimension $(r+s) \times (r+s)$ on the top left corner. Note that the matrix
$$N_i^{\ast(1)}=\begin{array}{|c|}
\hline
N_j^{(1)} \\ \hline
N_\ell^{(1)} \\ \hline
\end{array}$$
is an $(r+s)\times t$ matrix consisting of two matrices in row echelon form on top of each other.

To obtain $N_i$ from $N_i^\ast$, we perform row operations on $N_i^{\ast(1)}$. To preserve congruence, each row operation must be followed  by the corresponding column operation in $N_i^{\ast(1) T}$, but of course we need not actually perform the opereation. The goal is to insert the rows labeled by $V_1(N_\ell)$ (the \emph{right rows}) into the matrix $N_j^{(1)}$ labeled by the \emph{left rows} to produce a single matrix in row echelon form. We do this as follows. While there is a pivot $\alpha_c$ of a right row $w$ that lies in the same column $c$ as the pivot $\beta_c$ of $v$, where $v$ is a left row or a different right row, we use $v$ to eliminate the pivot of $w$ by performing the operations
\begin{equation}\label{eq_join}
R_w \leftarrow R_w-\frac{\alpha_c}{\beta_c} R_v, C_w \leftarrow C_w-\frac{\alpha_c}{\beta_c} C_v.
\end{equation}
After this loop has ended, the matrix $N_i^{(1)}$ is obtained from $N_j^{(1)}$ by simply inserting any right rows that still have a pivot in their proper position, in the sense that the final matrix $N_i^{(1)}$ is in row echelon form. Each row and column keeps its label throughout. If $Z_i$ denotes the set of labels of the right rows that became zero vectors while performing the above calculations, the algorithm appends diagonal elements $(v,0)$ for all $v\in Z_i$ to the global array and does not add their rows to $N_i^{(1)}$. This produces the box
\begin{equation}
\label{eq:formofN}
N_i = \begin{array}{|c|c|c|}
\hline
\mathbf{0}_{k' \times k'} & N_i^{(1)} \\ \hline
 N_i^{(1)T} &  N_i^{(2)}\\ \hline
\end{array},
\end{equation}
where $k'=r+s-|Z_i|$, $k''=t$ and $N_i^{(1)}$ is a matrix of dimension $k' \times k''$ in row echelon form and $N_i^{(2)}=N_i^{\ast(2)}$. Since every row of $N_i^{(1)}$ has a pivot, we have $k' \leq k''$.

\begin{figure}[h]
{\tt
\begin{tabbing}
aaa\=aaa\=aaa\=aaa\=aaa\=aaa\=aaa\=aaa\= \kill
     \> \texttt{JoinBox}($B_i,N_j,N_\ell$)\\
     \> {\bf input}: $B_i$ and boxes $N_j,N_\ell$ produced by the two children of $i$ \\
     \> {\bf output}: a box $N_i=(N_i^{(1)},N_i^{(2)})$ and a set of diagonal elements \\
     \> \> define $N_i^{*}$ as in~\eqref{eq:formofNast} \\
     \> \> do row operations on $N_i^{\ast(1)}$ as in~\eqref{eq_join}  to achieve row echelon form\\
     \> \> for each zero row of $N_i^{\ast(1)}$ (labeled by a vertex $u$), output $(u,0)$\\
     \> \> $N_i^{(1)}$ is $N_i^{*(1)}$ with zero rows removed as in~\eqref{eq:formofN}\\
\end{tabbing}}
\caption{Procedure \texttt{JoinBox}.}\label{fig:join}
\end{figure}

\vspace{10pt}

\noindent \textbf{ForgetBox.}  Assume that $i$ forgets vertex $v$ and let $j$ be its
child, so that $B_i=B_j \setminus \{v\}$. This procedure starts with $N_j$ and produces a new matrix $N_i$ where the row associated with $v$ becomes of type-i or is diagonalized. First it defines a new matrix $N_i^\ast$ from the box $N_j$, where, for all $u \in B_j$,
the entries $uv$ and $vu$ are replaced by $N_j^{(2)}[u,v]+m_{uv}$, while the other entries remain unchanged.

\begin{figure}[h]
{\tt
\begin{tabbing}
aaa\=aaa\=aaa\=aaa\=aaa\=aaa\=aaa\=aaa\= \kill
     \> \texttt{ForgetBox}($B_i,v,N_j$)\\
     \> {\bf input}: $B_i$, a vertex $v$ and matrix $N_j$ transmitted by the child of $i$ \\
     \> {\bf output}: a box $N_i=(N_i^{(1)},N_i^{(2)})$ and a set of diagonal elements \\
     \> \> $N_i^{*(1)}=N_j^{(1)}$,~$N_i^{*(2)}[u,w]=N_j^{(2)}[u,w]$, for all $u,w \in B_j$ with $v \notin\{u,w\}$\\
      \> \>$N_i^{*(2)}[u,v]=N_i^{*(2)}[v,u]=N_j^{(2)}[u,v]+m_{uv}$, for all $u \in B_j$\\
     \> \> Perform row/column exchange so that $N_i$ has the form of (\ref{eq:formofNast2})\\
     \> \> If $\mathbf{x}_v$ is empty or 0 then\\
     \> \> \> If $\mathbf{y}_v$ is empty or 0 then // \textit{Subcase 1(a)}\\
     \> \> \> \> append $(v, d_v)$ to the global array and remove row $v$ from $N_i$\\
     \> \> \> else //  Here $\{\mathbf{y}_v \not = 0\}$ \\
     \> \> \> \> If $d_v \not = 0 $ then // \textit{Subcase 1(b)} \\
     \> \> \> \>  \> use $d_v$ to diagonalize row/column $v$ as in~\eqref{eq:1b}\\
     \> \> \> \>  \> append $(v, d_v)$ to the global array and remove row $v$ from $N_i$\\
     \> \> \> \>  else // Here $d_v=0$ // \textit{Subcase 1(c)}  \\
     \> \> \> \> \> do row and column operations as in~\eqref{eq:1c} \\
     \> \> \> \>  \> If a zero row is obtained then\\
     \> \> \> \> \> \>  append $(v,0)$ to the global array and remove row $v$ from $N_i$\\
     \> \> \> \> \>  else insert row $v$ into $N_i^{(1)}$\\
     \> \> else // Here $\mathbf{x}_v \not = 0$  // \textit{Case 2} \\
     \> \> \> use  operations as in (\ref{eq_case21})-(\ref{eq_case23}) to diagonalize  rows
$u$ and $v$ \\
\> \> \> append $(v,d_v),(u, d_u)$ to the global array and remove rows $v,u$ from $N_i$.\\
\end{tabbing}}
\caption{Procedure \texttt{ForgetBox}.}\label{fig:forget}
\end{figure}

We observe that the row corresponding to $v$ is type-ii in the box $N_j$. Thus, after introducing the entries of $M$, $v$ is implicitly associated with a row in $N_i^{*(2)}$. For convenience, we exchange rows and columns to look at $N_i^\ast$ in the following way\footnote{This is helpful for visualizing the operations, but this step is not crucial in an implementation of this procedure.}:
\begin{equation}
\label{eq:formofNast2}
N_i^\ast = \begin{array}{|c|c|c|}
\hline
d_v & \mathbf{x}_v & \mathbf{y}_v \\ \hline
\mathbf{x}_v^T & \mathbf{0}_{k' \times k'} & N_i^{\ast(1)} \\ \hline
\mathbf{y}_v^T&  N_i^{\ast(1)T} &  N_i^{\ast(2)}\\ \hline
\end{array}.
\end{equation}
Here, the first row and column represent the row and column in $N_i^\ast$ associated with $v$, while $N_i^{\ast(1)}$ and $N_i^{\ast(2)}$ determine the entries $uw$ such that $u$ has type-i and $w \in B_i$, and such that $u,w \in B_i$, respectively. In particular $ \mathbf{x}_v$ and $ \mathbf{y}_v$ are row vectors of size $k'_j$ and $k''_j-1$, respectively.

Depending on $d_v$ and on the vectors $\mathbf{x}_v$ and $\mathbf{y}_v$, we proceed in different ways.

\vspace{5pt}

\noindent \emph{Case 1: $\mathbf{x}_v$ is empty or $\mathbf{x}_v=[0 \cdots 0]$.}

Case 1 happens when the box transmitted to node $i$ has no type-i rows and columns.

If $\mathbf{y}_v=[0 \cdots 0]$ (or $\mathbf{y}_v$ is empty), the row of $v$ is already diagonalized, we simply add $(v,d_v)$ to $D_i$ and remove the row and column associated with $v$ from $N_i^\ast$ to produce $N_i$. We refer to this as Subcase 1(a).

If $\mathbf{y}_v \neq [0 \cdots 0]$, there are again two options. In Subcase 1(b), we assume that $d_v \neq 0$ and we use $d_v$ to eliminate the nonzero entries in $y_v$ and diagonalize the row corresponding to $v$. For each element $u \in B_i$ such that the entry $\alpha_v$ of $y_v$ associated with $u$ is nonzero, we perform
\begin{equation}\label{eq:1b}
R_u \leftarrow R_u-\frac{\alpha_v}{d_v} R_v, C_u \leftarrow C_u-\frac{\alpha_v}{d_v} C_v.
\end{equation}
When all such entries have been eliminated, we append $(d_v,v)$  to the global array and we let $N_i$ be the remaining matrix. Observe that, in this case, $N_i^{(1)}=N_i^{\ast(1)}$, only the elements of $N_i^{\ast(2)}$ are modified to generate $N_i^{(2)}$.

If $\mathbf{y}_v \neq [0 \cdots 0]$ and $d_v=0$, we are in Subcase 1(c). The aim is to turn $v$ into a row of type-i. To do this, we need to insert $\mathbf{y}_v$ into the matrix  $N_i^{\ast(1)}$ in a way that the resulting matrix is in row echelon form. Note that this may be done by only adding multiples of rows of $V( N_i^{\ast(1)})$ to the row associated with $v$. At each step, if the pivot $\alpha_j$ of the (current) row associated with $v$ is in the same position as the pivot $\beta_j$ of $R_u$, the row associated with vertex $u$ already in $N_i^{\ast(1)}$, we use $R_u$ to eliminate the pivot of $R_v$:
\begin{equation}\label{eq:1c}
R_v \leftarrow R_v-\frac{\alpha_j}{\beta_j} R_u, C_v \leftarrow C_v-\frac{\alpha_j}{\beta_j} C_u.
\end{equation}
This is done until the pivot of the row associated with $v$ may not be cancelled by pivots of other rows, in which case the row associated with $v$ may be inserted in the matrix (to produce the matrix $N_i^{(1)}$), or until the row associated with $v$ becomes a zero row, in which case $(v,0)$ is added to $D_i$ and the row and column associated with $v$ are removed from $N_i^\ast$ to produce $N_i$.

\vspace{5pt}

\noindent \emph{Case 2: $\mathbf{x}_v$ is nonempty and $\mathbf{x}_v \neq [0 \cdots 0]$.}

Let $u$ be the vertex associated with the rightmost nonzero entry of $x_v$. Let $\alpha_j$ be this entry. We use this element to eliminate all the other nonzero entries in $x_v$, from right to left. Let $w$ be the vertex associated with an entry $\alpha_\ell \neq 0$. We perform
\begin{equation}\label{eq_case21}
R_w \leftarrow R_w - \frac{\alpha_\ell}{\alpha_j} R_u, \hspace{10pt} C_w \leftarrow C_w - \frac{\alpha_\ell}{\alpha_j} C_u.
\end{equation}
A crucial fact is that the choice of $u$ ensures that, even though these operations modify $N_i^{\ast(1)}$, the new matrix is still in row echelon form and has the same pivots as $N_i^{\ast(1)}$.
If $d_v \neq 0$, we still use $R_u$ to eliminate this element:
\begin{equation}\label{eq_case22}
R_v \leftarrow R_v - \frac{d_v}{2\alpha_j} R_u, \hspace{10pt} C_v \leftarrow C_v - \frac{d_v}{2\alpha_j} C_u.
\end{equation}
At this point, the only nonzero entries in the $(k'+1) \times (k'+1)$ left upper corner of the matrix obtained after performing these operations are in positions $uv$ and $vu$ (and are equal to $\alpha_j$). We perform the operations
 \begin{eqnarray}\label{eq_case23}
&&R_u  \leftarrow  R_u + \frac{1}{2}  R_v,  \hspace{3pt} C_u \leftarrow  C_u + \frac{1}{2}  C_v,  \hspace{3pt}  R_v  \leftarrow  R_v - R_u,  \hspace{3pt} C_v  \leftarrow  C_v - C_u.
\end{eqnarray}
The relevant entries of the matrix are modified as follows:
\begin{equation}
\label{eq:twobytwotrick}
\begin{pmatrix}
0 & \alpha_j \\
\alpha_j & 0
\end{pmatrix}
\rightarrow
\begin{pmatrix}
0 & \alpha_j \\
\alpha_j & \alpha_j
\end{pmatrix}
\rightarrow
\begin{pmatrix}
-\alpha_j & 0\\
0 & \alpha_j
\end{pmatrix}.
\end{equation}
We are now in the position to use the diagonal elements to diagonalize the rows associated with $v$ and $u$, as was done in Case 1, when $x_v=[0,\ldots,0]$ and $d_v \neq 0$. At the end of the step, we add $(v,-\alpha_u)$ and $(u,\alpha_u)$ to $D_i$. During this step, no pivots in $N_i^{\ast(1)}$ are modified, except for the pivot of the row associated with $u$, which is diagonalized during the step.

\section{Example}
\label{sec:example}

In this section, we illustrate how the algorithm of the previous section acts on a concrete example. Consider the symmetric matrix
$$M=(m_{ij})=
\left( \begin {array}{rrrrrr}  0&0&2&-1&0&0\\ 0&0&0&1&0&0\\ 2&0&1&3&2&0\\ -1&1&3&1&0&-1\\
 0&0&2&0&1&-1\\ 0&0&0&-1&-1&1
 \end {array}
 \right),
$$
whose graph is depicted in Figure~\ref{figure:ntc}, along with a nice tree decomposition $\mathcal{T}$. We wish to find a diagonal matrix that is congruent to $M$ using $\mathcal{T}$ as the input. We start with the following bottom-up ordering of the nodes of $\mathcal{T}$: nodes $1$ to $4$ are on the top branch, nodes $5$ to $8$ are on the bottom branch and nodes $9$ to $11$ go from the node of type join to the root of the tree.

For $i=1$, the node is a leaf with bag $B_1=\{ 1, 2, 4\}$ producing a box $N_1=\mathbf{0}_{3 \times 3}$ such that $N_1^{(1)}$ is empty and $N_1^{(2)}=N_1$ has rows labeled by the elements of $B_1$. Node $i=2$ forgets vertex $v=2$. The first step is to produce $N_2^\ast$ as in~\eqref{eq:formofNast2}. This requires updating the entries $uv$ and $vu$ such that $u \in B_1$ by summing them to the corresponding entries of the original matrix and exchanging rows and columns so that $v$ appears as the first row and column:
 $$\begin {array}{l} \texttt{2}\\ \texttt{1}\\\texttt{4}\end{array}\left( \begin {array}{r;{.5pt/1pt}rr} 0&0&1\\\hdashline[.5pt/1pt]
 0&0&0\\1&0&0\end {array} \right).
$$
The numbers on the left of each row denote their labels. When applying~\texttt{ForgetBox}, we have $d_v=0$, $x_v$ empty and $y_v=(0,1)$. This means that we are in Subcase 1(c), whose aim is to turn $v$ into a row of type-i. Since $N_2^\ast$ does not have any rows of type-i, no operations are performed and the node transmits the box
 $$N_2=\begin {array}{l} \texttt{2}\\ \texttt{1}\\\texttt{4}\end{array}\left( \begin {array}{r;{.5pt/1pt}rr} 0&0&1\\\hdashline[.5pt/1pt]
 0&0&0\\1&0&0\end {array} \right),
$$
where
$$N_2^{(1)}= \begin {array}{l}\texttt{2} \end{array}
\left(\begin {array}{cc}
0 & 1
\end{array}
\right)
\textrm{ and } N_2^{(2)}=\begin {array}{l}\texttt{1}\\\texttt{4}\end{array}\left( \begin {array}{rr}
 0&0\\0&0\end {array} \right),$$
to its parent. Node $i=3$ introduces vertex $3$. This consists in inserting a new zero row and column in $N_2$ and the corresponding zero column to $N_1$. This produces the box
 $$N_2=\begin {array}{l} \texttt{2}\\ \texttt{1}\\ \texttt{3}   \\ \texttt{4}\end{array}\left( \begin {array}{r;{.5pt/1pt}rrr} 0&0&0&1\\\hdashline[.5pt/1pt]
 0&0&0&0\\ 0&0&0&0\\1&0&0&0\end {array} \right).
$$
Note that the new row is inserted so that the labels of type-ii vertices are in increasing order. Node $i=4$ forgets vertex 1. We start by writing $N_4^\ast$ as in~\eqref{eq:formofNast2}:
$$N_{4}^*= \begin {array}{l} \texttt{1}\\ \texttt{2}\\\texttt{3}\\\texttt{4}\end{array} \left( \begin {array}{r;{.5pt/1pt}r;{.5pt/1pt}rr} 0&0&2&-1
\\\hdashline[.5pt/1pt] 0&0&0&1\\ \hdashline[.5pt/1pt] 2&0&0&0\\ -1&1&0&0\end
{array} \right).$$
This means that $v=1$, $d_v=0$, $x_v=(0)$ and $y_v=(2,-1)$. We are again in Subcase 1(c), and since the first nonzero element of $y_v$ is not in the same position as the pivot of the single row of type-i, the row associated with $v$ may simply be inserted as a type-i row to produce the box
 $$N_4=\begin {array}{l} \texttt{1} \\ \texttt{2}\\ \texttt{3}   \\ \texttt{4}\end{array}\left( \begin {array}{rr;{.5pt/1pt}rr} 0&0&2&-1\\  0&0&0&1 \\ \hdashline[.5pt/1pt]
 2&0&0&0\\-1&1&0&0\end {array} \right).
$$
Note that we have two type-i rows and two type-ii rows. Moreover, the matrix $N_4^{(1)}$ is upper triangular, and therefore is in row echelon form.

Node $i=5$ is again a leaf, with bag $B_5=\{ 3,5,6\}$, so it produces the box $N_5=\mathbf{0}_{3 \times 3}$ such that $N_5^{(1)}$ is empty and $N_5^{(2)}=N_5$ has rows labeled by the elements of $B_5$. Node $i=6$ forgets vertex 5. The first step is to produce $N_6^\ast$ as in~\eqref{eq:formofNast2}:
 $$N_6^\ast=\begin {array}{l} \texttt{5}\\ \texttt{3}\\\texttt{6}\end{array}\left( \begin {array}{r;{.5pt/1pt}rr} 1&2&-1\\\hdashline[.5pt/1pt]
 2&0&0\\-1&0&0\end {array} \right).
$$
Here $v=5$, $d_v=1$, $x_v$ is empty and $y_v=(2,-1)$. This directs us to Subcase 1(b) of \texttt{ForgetBox}, where we use $d_v$ to eliminate the nonzero elements of $y_v$. We perform the operations\footnote{The labels of the operations refer to the vertices that label each row and column of the box.} $R_3 \leftarrow R_3 - 2 R_5$, $C_3 \leftarrow C_3 - 2 C_5$, $R_6 \leftarrow R_6 + R_5$ and $C_6 \leftarrow C_6 + C_5$, in that order, which leads to the matrix
 $$\begin {array}{l} \texttt{5}\\ \texttt{3}\\\texttt{6}\end{array}\left( \begin {array}{r;{.5pt/1pt}rr} 1&0&0\\\hdashline[.5pt/1pt]
 0&-4&2\\0&2&-1\end {array} \right).
$$
The algorithm appends $(5,1)$ to the previously empty global array and node $i=6$ produces the box
 $$N_6=\begin {array}{l} \texttt{3}\\\texttt{6}\end{array}\left( \begin {array}{rr} -4&2\\2&-1\end {array} \right),
$$
where all rows have type-ii. Node $i=7$ receives $N_6$ and introduces vertex 4, which leads to new row and column of type-ii. This gives the box:
 $$N_7=\begin {array}{l} \texttt{3}\\ \texttt{4}\\ \texttt{6}\end{array}\left( \begin {array}{rrr} -4&0&2\\0&0&0\\ 2& 0&-1\end {array} \right),
$$
where all rows have type-ii. Node $i=8$ forgets vertex $v=6$, and we start by writing $N_8^\ast$ as in~\eqref{eq:formofNast2}:
 $$N_8^\ast=\begin {array}{l} \texttt{6}\\ \texttt{3}\\\texttt{4}\end{array}\left( \begin {array}{r;{.5pt/1pt}rr} 0&2&-1\\\hdashline[.5pt/1pt]
 2&-4&0\\-1&0&0\end {array} \right).
$$
As a consequence, $d_v=0$, $x_v$ is empty and $y_v=(2,-1)$. As before, we are in Subcase 1(c) and no operations are required to produce the box
 $$N_8=\begin {array}{l} \texttt{6}\\ \texttt{3}\\\texttt{4}\end{array}\left( \begin {array}{r;{.5pt/1pt}rr} 0&2&-1\\\hdashline[.5pt/1pt]
 2&-4&0\\-1&0&0\end {array} \right),
$$
where the first row has type-i and the other two rows have type-ii.

Node $i=9$ is a join. We apply \texttt{JoinBox} with input matrices $N_4$ and $N_{8}$. This first produces a matrix $N_9^{*}$ as in \eqref{eq:formofNast},
where the rows of type-i coming from the different branches are stacked on top of each other and the matrices corresponding to rows of type-ii are added to each other:
 $$N_9^{\ast}=\begin {array}{l}  \texttt{1} \\  \texttt{2}\\ \texttt{6}\\ \texttt{3}\\\texttt{4}\end{array}\left( \begin {array}{rr;{.5pt/1pt}r;{.5pt/1pt}rr} 0&0&0&2&-1 \\ 0&0&0&0&1\\ \hdashline[.5pt/1pt]
0&0&0&2&-1\\ \hdashline[.5pt/1pt]  2&0&2&-4&0 \\ -1&1&-1&0&0\end {array} \right).
$$
\texttt{JoinBox} instructs us to merge the single type-i row from the right branch, labeled by $w=6$, into the matrix labeled by the left rows in a way that preserves row echelon form, if possible. The pivot of row $w$ is the same as the pivot of the left row labeled by $v=1$, and the algorithm performs the operations $R_6 \leftarrow R_6-R_1$ and $C_6 \leftarrow C_6-C_1$, leading to
 $$\begin {array}{l}  \texttt{1} \\  \texttt{2}\\ \texttt{6}\\ \texttt{3}\\\texttt{4}\end{array}\left( \begin {array}{rr;{.5pt/1pt}r;{.5pt/1pt}rr} 0&0&0&2&-1 \\ 0&0&0&0&1\\ \hdashline[.5pt/1pt]
0&0&0&0&0\\ \hdashline[.5pt/1pt]  2&0&0&-4&0 \\ -1&1&0&0&0\end {array} \right).
$$
The row and column of $w$ became 0, so that the algorithm appends $(6,0)$ to the global array and transmits the following box to the node's parent:
 $$N_9=\begin {array}{l}  \texttt{1} \\  \texttt{2}\\ \texttt{3}\\\texttt{4}\end{array}\left( \begin {array}{rr;{.5pt/1pt}rr} 0&0&2&-1 \\ 0&0&0&1\\ \hdashline[.5pt/1pt]
  2&0&-4&0 \\ -1&1&0&0\end {array} \right).
$$
Note that two rows have type-i and two rows have type-ii. Node $i=10$ forgets vertex 4. The first step is to create $N_{10}^*$ as in~\eqref{eq:formofNast2} (note that only the entries involving $v=4$ and $u \in \{3,4\}$, which are the two vertices of type-ii):
$$N_{10}^*= \begin {array}{l} \texttt{4}\\ \texttt{1}\\\texttt{2}\\\texttt{3}\end{array} \left( \begin {array}{r;{.5pt/1pt}rr;{.5pt/1pt}r} 1&-1&1&3
\\\hdashline[.5pt/1pt] -1&0&0&2\\ 1&0&0&0\\  \hdashline[.5pt/1pt] 3&2&0&-4\end
{array} \right).$$
Here $v=4$, $d_v=1$, $x_v=(-1,1)$ and $y_v=(3)$. This leads to Case 2 in \texttt{ForgetBox}. The vertex associated with the rightmost nonzero entry of $x_v$ is $u=2$. Operations as in~\eqref{eq_case21} are performed, namely $R_1 \leftarrow R_2 + R_1$ and $C_1 \leftarrow C_2 + C_1$. This produces the matrix
$$\begin {array}{l} \texttt{4}\\ \texttt{1}\\\texttt{2}\\\texttt{3}\end{array} \left( \begin {array}{r;{.5pt/1pt}rr;{.5pt/1pt}r} 1&0&1&3
\\\hdashline[.5pt/1pt] 0&0&0&2\\ 1&0&0&0\\  \hdashline[.5pt/1pt] 3&2&0&-4\end
{array} \right).$$
Since $d_v \neq 0$, we perform operations as in~\eqref{eq_case22}, namely $R_4 \leftarrow R_4 - \frac{1}{2} R_2$ and $C_4 \leftarrow C_4 - \frac{1}{2} C_2$, reaching
$$\begin {array}{l} \texttt{4}\\ \texttt{1}\\\texttt{2}\\\texttt{3}\end{array} \left( \begin {array}{r;{.5pt/1pt}rr;{.5pt/1pt}r} 0&0&1&3
\\\hdashline[.5pt/1pt] 0&0&0&2\\ 1&0&0&0\\  \hdashline[.5pt/1pt] 3&2&0&-4\end
{array} \right).$$
Next the algorithm performs operations as in~\eqref{eq_case23}, namely $R_2 \leftarrow R_2 +\frac{1}{2} R_4$, $C_2 \leftarrow C_2 +\frac{1}{2} C_4$, $R_4 \leftarrow R_4 - R_2$ and $C_4 \leftarrow C_4 - C_2$. This produces
$$\begin {array}{l} \texttt{4}\\ \texttt{1}\\\texttt{2}\\\texttt{3}\end{array} \left( \begin {array}{r;{.5pt/1pt}rr;{.5pt/1pt}r} -1&0&0&3/2
\\\hdashline[.5pt/1pt] 0&0&0&2\\ 0&0&1&3/2\\  \hdashline[.5pt/1pt] 3/2&2&3/2&-4\end
{array} \right).$$
To conclude this step, the algorithm performs row and column operations to diagonalize the rows associated with $v=4$ and $u=2$, namely
$R_3 \leftarrow R_3 +\frac{3}{2} R_4$, $C_3 \leftarrow C_3 +\frac{3}{2} C_4$, $R_3 \leftarrow R_3 -\frac{3}{2} R_2$, $C_3 \leftarrow C_3 -\frac{3}{2} C_2$,
leading to
$$\begin {array}{l} \texttt{4}\\ \texttt{1}\\\texttt{2}\\\texttt{3}\end{array} \left( \begin {array}{r;{.5pt/1pt}rr;{.5pt/1pt}r} -1&0&0&0
\\\hdashline[.5pt/1pt] 0&0&0&2\\ 0&0&1&0\\  \hdashline[.5pt/1pt] 0&2&0&-4\end
{array} \right).$$
The algorithm appends the pairs $(4,-1)$ and $(2,1)$ to the global array and transmits the box
$$N_{10}=\begin {array}{l} \texttt{1}\\\texttt{3}\end{array} \left( \begin {array}{r;{.5pt/1pt}r} 0&2
\\\hdashline[.5pt/1pt] 2&-4\end
{array} \right)
$$
to its parent. The root node $i=11$ forgets vertex $v=3$. It first produces
$$N_{11}=\begin {array}{l} \texttt{3}\\\texttt{1}\end{array} \left( \begin {array}{r;{.5pt/1pt}r} -3&2
\\\hdashline[.5pt/1pt] 2&0\end
{array} \right),
$$
so that $d_v=-3$, $x_v=(2)$ and $y_v$ is empty. We are in case 2, and $x_v$ has a single nonzero entry, so that $u=1$. The operations in~\eqref{eq_case22}  are  $R_3 \leftarrow R_3 + \frac{3}{4} R_1$ and $C_3 \leftarrow C_3 + \frac{3}{4}   C_1$,  while the operations in~\eqref{eq_case23} are $R_1 \leftarrow R_1 +\frac{1}{2} R_3$, $C_1 \leftarrow C_1 +\frac{1}{2} C_3$, $R_3 \leftarrow R_3 - R_1$ and $C_3 \leftarrow C_3 - C_1$, leading to
$$\begin {array}{l} \texttt{3}\\\texttt{1}\end{array} \left( \begin {array}{r;{.5pt/1pt}r} -2&0
\\\hdashline[.5pt/1pt] 0&2\end
{array} \right).
$$
The algorithm adds $(3,-2)$ and $(1,2)$ to the global array and finishes. The global array produced by the algorithm is $D=\{(5,1),(6,0),(4,-1),(2,1),(3,-2),(1,2)\}$, meaning that $M$ is congruent to the diagonal matrix
$$D=
\left( \begin {array}{rrrrrr}  2&0&0&0&0&0\\ 0&1&0&0&0&0\\ 0&0&-2&0&0&0 \\ 0&0&0&-1&0&0 \\
 0&0&0&0&1&0 \\  0&0&0&0&0&0 \end {array}
 \right).
$$
This implies that $\det(M)=0$, $\rank(M)=5$ and $\inertia(M)=(3,2,1)$.

\section{Correctness of the algorithm} \label{sec:proof}

In this section, we prove Theorem~\ref{main_theorem}.  We start with some additional terminology that will be useful in our proof.

As the algorithm proceeds bottom-up on $\mathcal{T}$, we need to understand how it acts on some of its subtrees. Let $\mathcal{T}_i$ be the subtree of $\mathcal{T}$ rooted at node $i$, that is, $\mathcal{T}_i$ is the subtree of $\mathcal{T}$ induced by node $i$ and all of its descendants. Also, for any fixed $v \in V(G)$, let $\mathcal{T}(v)$ denote the subtree of $\mathcal{T}$ induced by all nodes whose bag contains $v$. We view $\mathcal{T}(v)$  as rooted at the child $j$ of the node $i$ that forgets $v$. The definition of tree decomposition ensures that, among all nodes whose bag contains $v$, $j$ is the closest to the root of $\mathcal{T}$.

Recall that the diagonal entries appended by the algorithm to the global array are pairs $(v,d_v)$, where $v$ is a vertex of $G$ (equivalently, a row of $M$) and $d_v$ is the diagonal entry associated with it. Let $D_i$ be the set of all pairs $(v,d_v)$ produced up to the end of Step $i$. Let $\pi_1(D_i)$ and $\pi_2(D_i)$ be the projections of $D_i$ onto their first and second coordinates, respectively, so that $\pi_1(D_i)$ will be the set of rows that have been diagonalized up to the end of step $i$ and $\pi_2(D_i)$ will be the (multi)set of diagonal elements found up to this step\footnote{We use the future tense to emphasize that we have yet to prove that the components of $D_i$ satisfy what is claimed.}.

To measure the effect of the operations performed by the algorithm in terms of the input matrix $M$, let $\tilde{M}(i)$ be the matrix that is congruent to $M$ obtained by performing the row and column operations performed by the algorithm up to step $i$. However, since the algorithm acts independently on the different branches of $\mathcal{T}$, for each node $i$, it is convenient to consider a matrix $M_i$ that captures only the information of $M$ that was used in the algorithm for nodes of the branch $\mathcal{T}_i$. It is defined as follows, where for clarity, we use the notation $A[u,v]$ to denote the entry $a_{uv}$ of a matrix $A$:
\begin{equation}\label{defmi}
M_i[u,v]=
\begin{cases}
m_{uv}, & \textrm{ if } \mathcal{T}(u) \cap \mathcal{T}_i \neq \emptyset,\mathcal{T}(v) \cap \mathcal{T}_i \neq \emptyset,~\{u,v\} \not\subseteq B_i,\\
0, & \textrm{ otherwise.}
\end{cases}
\end{equation}
The effect of different branches is merged at nodes of type Join. Note that, at the final step of the algorithm, we have $\mathcal{T}_m=\mathcal{T}$ and $M_m=M$.

Our technical lemmas do not refer to $\tilde{M}(i)$ and, instead, describe the relationship between $N_i$ and the diagonal elements produced in step $i$ with the matrices $M_i$ and $\tilde{M}_i$. To get started, we set $D_{-1}=D_0=\emptyset$ and $\mathcal{T}_0=\emptyset$. Let $N_0$ be the empty box and assume that no operations are performed before step 1. The matrices $M_0$ and $\tilde{M}_0$ are zero matrices of order $n$.

The lemma below states some intuitive facts about the algorithm. Part (a) mentions the simple fact that matrices obtained from one another by a sequence of pairs of the same elementary row and column operations are congruent. Part (b) establishes that any pairs added to the global array $D_i$ stay there, and that each row is diagonalized at most once. Part (c) states the important fact that a multiple of row $v$ can only be added to other rows $w$ through an elementary row operation during step $i$ if it becomes diagonalized or has type-i in this step.
\begin{lemma}\label{tech_lemma0}
The following facts hold for all $i \in \{0,\ldots,m\}$.
\begin{itemize}
\item[(a)] ~$\tilde{M}_i$ is a symmetric matrix congruent to $M_i$.

\item[(b)] ~$D_{i-1} \subseteq D_i$ and $\pi_1$ is injective over $D_i$.

\item[(c)] ~If a multiple of row (or column) $v$ has been added to another row (or column) $u$ in step $i$, then $v \in \pi_1(D_i\setminus D_{i-1}) \cup V_1(N_i)$.
\end{itemize}
\end{lemma}

The second lemma relates subtrees of type $\mathcal{T}(v)$ and $\mathcal{T}_i$ with vertex types. The first item says that if vertex $v$ appears in the bag of some node of $\mathcal{T}_i$, then $v$ is associated with a row of the box produced by $i$ or it has already been diagonalized. Conversely, part (b) says that if $v$ is associated with a row of the box produced by $i$ or if it has been diagonalized during step $i$, then it must appear in the bag of some node of $\mathcal{T}_i$. Part (c) implies that, if $v$ has been diagonalized up to step $i$ or has type-i in $N_i$, then $\mathcal{T}(v)$ is either fully contained or completely disjoint from $\mathcal{T}_i$. Part (d) confirms the intuitive fact that the concepts of being type-i, type-ii or diagonalized are mutually exclusive.
\begin{lemma}\label{tech_lemma}
For all $i \in \{0, \ldots, m\}$, $D_i$ and $N_i$ have the following additional properties.
\begin{itemize}

\item[(a)] ~Let $v$ be such that $\mathcal{T}(v) \cap \mathcal{T}_i \neq \emptyset$. Then $v \in V(N_i) \cup \pi_1(D_i)$.

\item[(b)]  ~If $v \in \pi_1(D_i\setminus D_{i-1}) \cup V(N_i)$, then $\mathcal{T}_i \cap \mathcal{T}(v) \neq \emptyset$.

\item[(c)] ~If $v \in \pi_1(D_i) \cup V_1(N_i)$ and $\mathcal{T}_i \cap \mathcal{T}(v) \neq \emptyset$, then $\mathcal{T}(v)$ is a subtree of $ \mathcal{T}_i$.

\item[(d)] ~The sets $\pi_1(D_i)$, $V_1(N_i)$ and $V_2(N_i)$ are mutually disjoint.

\end{itemize}
\end{lemma}

The third result relates the entries of the matrices $N_i$ and of the array $D$ produced by the algorithm with the entries of $\tilde{M}_i$.
\begin{lemma}\label{tech_lemma2}
For all $i \in \{0, \ldots, m\}$, $\tilde{M}_i$ has the following additional properties.
\begin{itemize}
\item[(a)] ~If $u,v \in B_i$, then the entry $uv$ of $\tilde{M}_i$ is equal to the entry $uv$ of $N_i^{(2)}$.

\item[(b)] ~If $v \in V_1(N_i)$, then the row (and column) associated with $v$ in $\tilde{M}_i$ satisfies the following. For any $u \in V(G)$, the entries $uv$ and $vu$ are equal to 0 if $u \notin B_i$ and are equal to the corresponding entries in $N_i$ if $u \in B_i$.

\item[(c)] ~If $(v,d_v) \in D_i$, $\mathcal{T}_i \cap \mathcal{T}(v) \neq \emptyset$, then the row (and column) associated with $v$ in $\tilde{M}_i$ consists of zeros, with the possible exception of the $v$th entry, which is equal to $d_v$.

\item[(d)] ~If $\mathcal{T}_i \cap \mathcal{T}(v) = \emptyset$, then the row (and column) associated with $v$ in $\tilde{M}_i$ are zero.

\end{itemize}
\end{lemma}

The proofs of Lemmas~\ref{lemma_Ni}, \ref{tech_lemma0},~\ref{tech_lemma} and~\ref{tech_lemma2} are by induction on $i$. Before presenting this inductive argument, we argue that they imply that algorithm \texttt{CongruentDiagonal} gives the correct output, namely a diagonal matrix $D$ that is congruent to $M$.

\begin{proof}[Proof of Theorem~\ref{main_theorem} (Correctness)]
By definition of $M_i$ in~\eqref{defmi} and the fact that $\mathcal{T}_m=\mathcal{T}$, we have $M_m=M$. By Lemma~\ref{tech_lemma0}(a), we have that $M$ is congruent to $\tilde{M}_m$, so it suffices to prove that $\tilde{M}_m$ is a diagonal matrix whose diagonal entries have been appended to the global array.

First we argue that $\pi_1(D_m)=V$. Let $v \in V$. Since the root $m$ of $\mathcal{T}$ has an empty bag, the box $N_m$ produced by it must be empty by Lemma~\ref{lemma_Ni}, items (a) and (c). Also, $\mathcal{T}(v) \cap \mathcal{T}_m = \mathcal{T}(v) \cap \mathcal{T}  \neq \emptyset$ and, by Lemma~\ref{tech_lemma}(a),  $v \in V(N_m) \cup \pi_1(D_m)=\pi_1(D_m)$, as claimed.

The result now follows from Lemma~\ref{tech_lemma2}(c), which implies that $\tilde{M}_m$ is a diagonal matrix such that, for all $v \in [n]$, the entry $vv$ is given by $\pi_2 \circ \pi_1^{-1}(v)=d_v$. These are the elements produced by \texttt{CongruentDiagonal}, as Lemma~\ref{tech_lemma0}(b) tells us that no elements originally appended to the global array at step $i$ may be altered in a later step.

\end{proof}

We will now prove the lemmas. Instead of proving them one at a time, we will prove that all of their statements hold by induction. In fact, it is easy to see that the steps of the algorithm were designed so that the boxes satisfy the properties in the statement of Lemma~\ref{lemma_Ni}. What needs to be proved is that these boxes really give the information that is expected from them.

At the start of the algorithm, when $i=0$, $N_0$ is empty, $\mathcal{T}_0$ is empty, the matrices $M_0$ and $\tilde{M}_0$ are both zero matrices of order $n$, and $D_{-1}=D_0=\emptyset$. No operations have been performed. So Lemmas~\ref{tech_lemma0},~\ref{tech_lemma} and~\ref{tech_lemma2} hold trivially.
Fix $i \in \{1,\ldots,m\}$ and assume that the lemmas hold for all $j<i$. Our discussion will depend on the type of node $i$.

\vspace{5pt}

\noindent \textbf{Case 1 (Leaf).} If $i$ is a leaf, then $M_i$ is a zero matrix of order $n$. Recall that the box $N_i$ produced at this step has no rows of type-i, and $N_i^{(2)}$ is a zero matrix of order $|B_i|$ whose rows and columns are associated with the elements of $B_i$.

Concerning Lemma~\ref{tech_lemma0}(a), the algorithm does not perform any row or column operation during step $i$, so $\tilde{M}_i=M_i$ and (a) holds. Item (b) holds by the induction hypothesis because $D_i=D_{i-1}$. Item (c) holds by vacuity.

We next move to Lemma~\ref{tech_lemma}. Here $\mathcal{T}_i$ is a rooted tree with a single node, whose bag is $B_i$. So $\mathcal{T}(v) \cap \mathcal{T}_i \neq \emptyset$ if and only if $v \in B_i$, which means that $v \in V_2(N_i)=V(N_i)$. This implies (a) and (b), where the latter uses $D_i\setminus D_{i-1}=\emptyset$. Item (c) holds by vacuity, as no vertex $v$ satisfies the hypothesis. For part (d), we know that $V_1(N_i)=\emptyset$, so it suffices to consider $\pi_1(D_i)$ and $V_2(N_i)=B_i$. Fix $v \in \pi_1(D_i)$. Since no elements are appended to the global array at step $i$ and $D_0=\emptyset$, let $1 \leq j<i$ such that $v \in \pi_1(D_j\setminus D_{j-1})$. By Lemma~\ref{tech_lemma}(b) and (c), $\mathcal{T}(v)$ is a subtree of $\mathcal{T}_j$. Clearly, $\mathcal{T}_j \cap \mathcal{T}_i=\emptyset$, which implies that $v \notin B_i$, as required.

 We conclude with Lemma~\ref{tech_lemma2}. Items (a) and (d) hold because all entries of both $N_i$ and $M_i$ are zero. Items (b) and (c) hold by vacuity, as no vertex $v$ satisfies the hypothesis.

\vspace{5pt}

\noindent \textbf{Case 2 (Introduce).}  Assume that $i$ has type Introduce, $j$ is its child and $v^*$ denotes the vertex that has been introduced. Recall that the box $N_i$ produced at this step simply adds a new type-ii row (and column) whose entries are 0 to the box $N_j$ transmitted by node $j$. We claim that $M_i=M_j$. To see why this is true, consider the definition~\eqref{defmi}. On the one hand, since $B_j=B_i \setminus \{v^*\}$ the entries $uv$ must be the same if $u \neq v^*$ and $v \neq v^\ast$. On the other hand, we claim that $m_{uv^\ast}=0$ for any $u \in V$ such that $\mathcal{T}(u) \cap \mathcal{T}_i\neq \emptyset$ and $u \notin B_i$. The latter two conditions imply that $\mathcal{T}(u)$ is a subtree of $\mathcal{T}_j$ by Lemma~\ref{tech_lemma}(c), which holds for step $j$ by the induction hypothesis. Since $v^\ast$ is introduced at $i$, there cannot be a bag containing both $u$ and $v^\ast$, and $\{u,v^\ast\} \notin E(G)$ because $\mathcal{T}$ is a tree decomposition of $G$. This confirms that $M_i=M_j$.

We now consider each lemma separately. We start with Lemma~\ref{tech_lemma0}. Since the algorithm does not perform any row or column operation during step $i$ and $M_i=M_j$, we must have $\tilde{M}_i=\tilde{M}_j$ and (a) holds by the induction hypothesis. As no diagonal elements are produced at this step, $D_i=D_{i-1}$ and (b) holds by the induction hypothesis. Item (c) holds by vacuity.

Next consider Lemma~\ref{tech_lemma}. For part (a), let $v$ be such that $\mathcal{T}(v) \cap \mathcal{T}_i \neq \emptyset$. If $v=v^\ast$, then $v \in V(N_i)$. Otherwise $\mathcal{T}(v) \cap \mathcal{T}_j \neq \emptyset$, so that by the induction hypothesis
$$v \in V(N_j) \cup \pi_1(D_j) \stackrel{(I)}{\subseteq} V(N_i) \cup \pi_1(D_{i-1})=V(N_i) \cup \pi_1(D_{i}).$$
Note that in (I) we inductively used Lemma~\ref{tech_lemma0}(b) to conclude that $D_j\subseteq  D_{i-1}$.

For part (b), let $v$ be such that $v \in \pi_1(D_i \setminus D_{i-1}) \cup V(N_i)=V(N_i)$. Either $v=v^\ast \in B_i$, so $\mathcal{T}(v) \cap \mathcal{T}_i \neq \emptyset$, or $v \in V(N_j)$, so that  $\mathcal{T}(v) \cap \mathcal{T}_i \supseteq \mathcal{T}(v) \cap \mathcal{T}_j  \neq \emptyset$ by the induction hypothesis.

Before proving part (c), we prove part (d). We do this in detail because this type of argument will be used several times. By the induction hypothesis, $\pi_1(D_{j}) \cup V_1(N_j) \cup V_2(N_j)$ is a disjoint union. Since $V_1(N_i)=V_1(N_j)$ and $V_2(N_i)=V_2(N_j) \cup \{v^\ast\}$, it suffices to prove that $\pi_1(D_i\setminus D_j) \cap V(N_i) = \emptyset$ and that $v^\ast \notin V_1(N_i) \cap \pi_1(D_i)$. For the first part, fix $u \in \pi_1(D_i\setminus D_j)= \pi_1(D_{i-1}\setminus D_j)$. Let $\ell$, where $j<\ell<i$, such that $u \in \pi_1(D_\ell\setminus D_{\ell-1})$. By the induction hypothesis, Lemma~\ref{tech_lemma}(b) and (c) hold at step $\ell$, so $\mathcal{T}(u)$ is a subtree of $\mathcal{T}_\ell$. The trees $\mathcal{T}_\ell$ and $\mathcal{T}_i$ are disjoint because of the bottom-up ordering of the nodes of $\mathcal{T}$. By part (b), any $v \in V(N_i)$ is such that $\mathcal{T}(v) \cap \mathcal{T}_i \neq \emptyset$, establishing that $\pi_1(D_i\setminus D_j) \cap V(N_i) = \emptyset$. Showing that  $v^\ast \notin V_1(N_i) \cap \pi_1(D_i)$ is similar: given $u \in \pi_1(D_i)=\pi_1(D_{i-1})$, an argument as above would also give $\ell<i$  such that $\mathcal{T}(u)$ is a subtree of $\mathcal{T}_\ell$, so that $u \neq v^\ast \in B_i$. If $w \in  V_1(N_i)=V_1(N_j)$, Lemma~\ref{tech_lemma}(b) and (c) applied at step $j$ implies that  $\mathcal{T}(w)$ is a subtree of $\mathcal{T}_j$, and again $w \neq v^\ast$.

For part (c), let $v$ be such that $v \in \pi_1(D_i) \cup V_1(N_i)$ and $\mathcal{T}(v) \cap \mathcal{T}_i \neq \emptyset$. Since $v^\ast \in V_2(N_i)$ (and by part (d), which has already been proved), we must have $v\neq v^\ast$ and $\mathcal{T}(v) \cap \mathcal{T}_i \neq \emptyset$. This means that $\mathcal{T}(v) \cap \mathcal{T}_j \neq \emptyset$ and, by Lemma~\ref{tech_lemma}(a), $v \in V(N_j) \cup \pi_1(D_j)$.  Since $V_2(N_j) \subset V_2(N_i)$ and $v\notin V_2(N_i)$, we have $v \in V_1(N_j) \cup \pi_1(D_j)$, so that, by the induction hypothesis, $\mathcal{T}(v)$ is a subtree of $\mathcal{T}_j$, and hence of $\mathcal{T}_i$.

We finish with Lemma~\ref{tech_lemma2}. For part (a), fix $u,v \in B_i$. If $v^\ast \notin \{u,v\}$, we know that $u,v \in B_j$ and
$\tilde{M}_i[u,v]=\tilde{M}_j[u,v]=N_j^{(2)}[u,v]=N_i^{(2)}[u,v].$
The first and last equalities are due to the way step $i$ works, the second equality holds by the induction hypothesis.  Next assume that $v^\ast \in \{u,v\}$, say $v=v^\ast$. Since $\mathcal{T}_j \cap \mathcal{T}(v^\ast)=\emptyset$, an inductive application of Lemma~\ref{tech_lemma2}(d) tells us that the row and column associated with $v^\ast$ in $\tilde{M}_j$ are zero. The result now follows from $\tilde{M}_i=\tilde{M}_j$ and the fact that the row and column associated with $v^\ast$ in $N_i$ are zero.

For part (b), assume that $v \in V_1(N_i)= V_1(N_j)$. By the induction hypothesis, the row and column associated with $v$ in $\tilde{M}_j$ satisfy the following. For any $u \in V(G)$, the entries $uv$ and $vu$ are equal to 0 if $0 \notin B_j$ and are equal to the corresponding entries in $N_j$ if $u \in B_j$. Nothing changes at step $i$ except from the fact that $v^\ast  \notin B_j$, but $v^\ast  \in B_i$. The result holds because the row and column associated with $v^\ast$ in $N_i$ are zero.

For part (c), fix $v$ such that $(v,d_v) \in D_i$ and $\mathcal{T}_i \cap \mathcal{T}(v) \neq \emptyset$. By Lemma~\ref{tech_lemma}(d), $v \neq v^\ast$, so $\mathcal{T}_j \cap \mathcal{T}(v) \neq \emptyset$. We claim that $(v,d_v) \in D_j$, which implies the desired conclusion by the induction hypothesis and the fact that $\tilde{M}_i=\tilde{M}_j$. Assume for a contradiction that $(v,d_v) \notin D_j$, so that $(v,d_v) \in D_\ell \setminus D_{\ell-1}$, for some $\ell \in \{j+1,i-1\}$. By Lemma~\ref{tech_lemma}(b) and (c), applied at step $\ell$, we see that $\mathcal{T}(v)$ is a subtree of $\mathcal{T}_\ell$. This contradicts the assumption $\mathcal{T}_i \cap \mathcal{T}(v) \neq \emptyset$, as the bottom-up ordering of the nodes of $\mathcal{T}$ leads to $\mathcal{T}_\ell \cap \mathcal{T}_i=\emptyset$.

For part (d), fix $v$ such that $\mathcal{T}_i \cap \mathcal{T}(v) = \emptyset$. So $\mathcal{T}_j \cap \mathcal{T}(v) = \emptyset$ and the row and column associated with $v$ in $\tilde{M}_j$ are 0 by the induction hypothesis. The result follows from $\tilde{M}_i=\tilde{M}_j$.

\vspace{5pt}

\noindent \textbf{Case 3 (Join).} If $i$ has type Join, let $j$ and $\ell$ be its children and assume that $j<\ell$. Since the construction of $N_i$ from $N_j$ and $N_\ell$ is more involved, we shall not repeat it here. By Lemma~\ref{lemma_Ni}, we have $V_2(N_j)=V_2(N_\ell)=B_j=B_i$. Moreover, let $u \in V_1(N_j)$ and $v \in V_1(N_\ell)$. By the induction hypothesis, Lemma~\ref{tech_lemma}(b) and (c), applied at steps $j$ and $\ell$, tells us that $\mathcal{T}(u)$ and $\mathcal{T}(v)$ are subtrees of $\mathcal{T}_j$ and $\mathcal{T}_\ell$, respectively. Since $\mathcal{T}_j \cap \mathcal{T}_\ell=\emptyset$, we conclude that $u \neq v$ and therefore $V_1(N_j) \cap V_1(N_\ell)=\emptyset$, proving a claim made in the description of \texttt{JoinBox}. The next result describes how the matrix $M_i$ changes when two branches merge.

\begin{claim}\label{claim:joinbox}
The matrix $M_i$ in~\eqref{defmi} satisfies $M_i=M_j+M_\ell$. Moreover, if a multiple of row $v$ has been added to another row $u$ while processing a node in $\mathcal{T}_j$, then all entries in the row/column of $v$ in $M_\ell$ are zero. Similarly, if a multiple of row $v$ has been added to another row $u$ while processing a node in $\mathcal{T}_\ell$, then all entries in the row/column of $v$ in $M_j$ are zero.
\end{claim}

\begin{proof} Fix $u,v \in V$. If $\mathcal{T}(u) \cap \mathcal{T}_i=\emptyset$, then we obviously have $\mathcal{T}(u) \cap \mathcal{T}_j=\emptyset$ and $\mathcal{T}(u) \cap \mathcal{T}_\ell=\emptyset$, and therefore $M_i[u,v]=0=0+0=M_j[u,v]+M_\ell[u,v]$. The same holds if $\mathcal{T}(v) \cap \mathcal{T}_i=\emptyset$. This conclusion is also reached for $\{u,v\} \subseteq B_i=B_j=B_\ell$.

Next assume that $\mathcal{T}(u) \cap \mathcal{T}_i \neq \emptyset$, $\mathcal{T}(v) \cap \mathcal{T}_i \neq \emptyset$ and $\{u,v\} \notin B_i$, implying $M_i [u,v] = m_{uv}$. Suppose without loss of generality that $u \notin B_i$. This means that $\mathcal{T}(u) \cap \mathcal{T}_j \neq \emptyset$ or $\mathcal{T}(u) \cap \mathcal{T}_\ell \neq \emptyset$.

If $\mathcal{T}(u) \cap \mathcal{T}_j \neq \emptyset$, given that $u \notin B_j$, Lemma~\ref{tech_lemma}(a) guarantees that $u \in V_1(N_j) \cup \pi_1(D_j)$, so that $\mathcal{T}(u)$ is a subtree of $\mathcal{T}_j$ by part (c) of the same lemma. Since $\mathcal{T}_j \cap \mathcal{T}_\ell=\emptyset$, we immediately deduce that $M_\ell[u,v]=0$. If $\mathcal{T}(v) \cap \mathcal{T}_j \neq \emptyset$, we get $M_i[u,v]=m_{uv}=M_j[u,v]=M_j[u,v]+M_\ell[u,v]$, as required. If $\mathcal{T}(v) \cap \mathcal{T}_j = \emptyset$, then $M_i[u,v]=m_{uv}$, while in addition to $M_\ell[u,v]$, also $M_j[u,v]=0$. However, since $\mathcal{T}(v) \cap \mathcal{T}_i \neq \emptyset$, with the same arguments above, we see that $\mathcal{T}(v)$ is a subtree of $\mathcal{T}_\ell$. This leads to $\mathcal{T}(u) \cap \mathcal{T}(v)=\emptyset$, so that $m_{uv}=0$ because no bag contains both $u$ and $v$.

If $\mathcal{T}(u) \cap \mathcal{T}_\ell \neq \emptyset$, the same conclusion may be achieved replacing $j$ by $\ell$ in the above argument, establishing the first part of the claim.

Next assume that a multiple of row $v$ has been added to another row $u$ while processing node $j'$ in $\mathcal{T}_j$. By Lemma~\ref{tech_lemma0}(c), $v \in \pi_1(D_{j'}\setminus D_{j'-1}) \cup V_1(N_{j'})$. By Lemma~\ref{tech_lemma}(b) and (c), $\mathcal{T}(v)$ is a subtree of  $\mathcal{T}_{j'}$ and hence of  $\mathcal{T}_{j}$. This means that $\mathcal{T}(v) \cap \mathcal{T}_\ell=\emptyset$, so that all entries in the row/column of $v$ in $M_\ell$ are zero by~\eqref{defmi}. The analogous conclusion is valid if multiple of row $v$ is added to a row $u$ while processing node $\ell'$ in $\mathcal{T}_\ell$.
\end{proof}

We are now ready to prove that Lemmas~\ref{tech_lemma0},~\ref{tech_lemma} and~\ref{tech_lemma2} still hold after step $i$ is performed.

We start with Lemma~\ref{tech_lemma0}. Part (a) is a simple consequence that $\tilde{M}_i$ is obtained from $M_i$ by a sequence of elementary row and column operations, where each row operation is followed by the same column operation. For part (b), note that step $i$ does not remove any elements from the global array, so $D_{i-1} \subseteq D_i$. On the other hand, it may append some pairs $(u,0)$ to the global array. For any such pair $(u,0)$, $u$ started the step as a row of type-i of the right child $r \in \{j,\ell\}$ of $i$, that is, $u \in V_1(N_r)$. By Lemma~\ref{tech_lemma}(d), $u \notin \pi_1(D_r) \cup V_2(N_r)= \pi_1(D_r) \cup B_i$. Mimicking the argument used in the proof of Lemma~\ref{tech_lemma}(d) for introduce nodes, we derive that $u \notin \pi_1(D_{i-1}\setminus D_r)$. This ensures that $\pi_1$ is injective over $D_{i}$. To see that part (c) holds, note that a multiple of a row $v$ may only be added to another row in operations of the form~\eqref{eq_join}, so that $v$ started the step in $V_1(N_j) \cup V_1(N_\ell)$. At the end of the step it is either in $D_i$ or $V_1(N_i)$.

Next, we consider Lemma~\ref{tech_lemma}. For part (a), let $v$ such that $\mathcal{T}(v) \cap \mathcal{T}_i \neq \emptyset$. Since $B_i=B_j=B_\ell$, this means that $\mathcal{T}(v) \cap \mathcal{T}_j \neq \emptyset$ or $\mathcal{T}(v) \cap \mathcal{T}_\ell \neq \emptyset$, so that $v \in V(N_j) \cup V(N_\ell) \cup \pi_1(D_\ell) \subset V(N_i) \cup \pi_1(D_i)$. Note that we are using Lemma~\ref{tech_lemma0}(b) in this argument. For part (b), fix $v \in \pi_1(D_i \setminus D_{i-1}) \cup V(N_i)$. By the description of \texttt{JoinBox}, we have $v \in V(N_j) \cup V(N_\ell)$, thus by the induction hypothesis $\mathcal{T}_j \cap \mathcal{T}(v) \neq \emptyset$ or $\mathcal{T}_\ell \cap \mathcal{T}(v) \neq \emptyset$. In either case, $\mathcal{T}_i \cap \mathcal{T}(v) \neq \emptyset$, as required.

As for introduce nodes, we prove (d) before (c). Since $V_1(N_i) \subset V_1(N_j) \cup V_1(N_\ell)$ and $V_2(N_i)=V_2(N_j)=V_2(N_\ell)$, the sets $V_1(N_i)$ and $V_2(N_i)$ are disjoint by the induction hypothesis. Induction also implies that $V(N_j)$ and $V(N_\ell)$ are disjoint from $D_j$ and $D_\ell$, respectively. The proof that $\pi_1(D_{i-1} \setminus D_j)$ and $\pi_1(D_{i-1} \setminus D_\ell)$ are disjoint from $V(N_j)$ and $V(N_\ell)$, respectively, is the same used for nodes of type introduce. Finally, it is clear that any rows and columns in $D_i \setminus D_{i-1}$ have been removed from $N_i^{\ast}$  to produce $N_i$, establishing (d). For part (c), fix $v \in \pi_1(D_{i}) \cup V_1(N_i)$ such that $\mathcal{T}_i \cap \mathcal{T}(v) \neq \emptyset$. By the description of \texttt{JoinBox} and arguments as above, we conclude that $v \in \pi_1(D_{\ell}) \cup V_1(N_j) \cup V_1(N_\ell)$, so that it is a subtree of $\mathcal{T}_j$ or $\mathcal{T}_\ell$ by the induction hypothesis, and hence of $\mathcal{T}_i$, as required.

We finally consider Lemma~\ref{tech_lemma2}. Before addressing each item individually, we explain the relationship between the entries of $N_i^\ast$ at the beginning of the step and a matrix that is congruent to $M_i$. The first part of Claim~\ref{claim:joinbox} tells us that $M_i=M_j+M_\ell$. The second part ensures that, if $\hat{M}_i$ denotes the matrix obtained from $M_i$ by performing the elementary row and column operations performed for nodes in $\mathcal{T}_j$ and $\mathcal{T}_\ell$, in the order in which they have been performed by the algorithm, we have $\hat{M}_i=\tilde{M}_j+\tilde{M}_\ell$. By Lemma~\ref{tech_lemma2}(a) applied at steps $j$ and $\ell$, we conclude that, for any $u,v \in B_i$, the entries $uv$ and $vu$ in $\hat{M}_i$ are equal to $N_j^{(2)}[u,v]+N_\ell^{(2)}[u,v]$.
Also, since $V_1(N_j) \cap V(N_\ell)=V_1(N_\ell) \cap V(N_j)=\emptyset$, Lemma~\ref{tech_lemma2}(b) applied at steps $j$ and $\ell$ ensures that any entries $uv$ and $vu$ of $\hat{M}_i$ where $u \in V_1(N_j)$ and $v \in V_1(N_\ell)$, or vice-versa, are equal to 0. Finally, if $u,v \in V(N_j)$ and $|\{u,v\} \cap B_i|\leq 1$, then the entries $uv$ and $vu$ in $\hat{M}_i$ and $\tilde{M}_j$ are the same by Claim~\ref{claim:joinbox}, and are equal to the corresponding entry in $N_j^{(0)}$, $N_j^{(1)}$ or $N_j^{(1)T}$ by Lemma~\ref{tech_lemma2}(b) applied at step $j$. The case $u,v \in V(N_\ell)$ and $|\{u,v\} \cap B_i|\leq 1$ is analogous. As a consequence, the entries that we see in~\eqref{eq:formofNast} are precisely the principal submatrix of $\hat{M}_i$ induced by rows in $V_1(N_j) \cup V_1(N_\ell) \cup B_i$.

Part (a) of Lemma~\ref{tech_lemma2} now follows for step $i$ because the operations performed during this step to $N_i^{\ast}$ (in order to produce $N_i$) and to $\hat{M}_i$ (in order to produce $\tilde{M}_i$) involve the same rows and columns, in the same order, so that all entries are affected in the same way.

For part (b), let $v \in V_1(N_i)$. Looking at the representation in~\eqref{eq:formofNast}, it is easy to see that operations in~\eqref{eq_join} do not modify the submatrix $\mathbf{0}_{(r+s)\times (r+s)}$ on the upper left corner. This shows that, if $u \in V_1(N_i)$, the entries $uv$ and $vu$ in $N_i$, and hence in $\tilde{M}_i$, must be 0. The same argument of part (a) ensures that, if $u \in B_i$, the entries $uv$ and $vu$ in $N_i$ and $\tilde{M}_i$ coincide. Finally, assume that $u \notin V(N_i)$. If $u \in V_1(N_j) \cup V_1(N_\ell)$, then $(u,0) \in D_i\setminus D_{i-1}$, and the entries $uv$ and $vu$ are again 0 because the same operations have been performed on $N_i^\ast$ and $\hat{M}_i$. To conclude the proof of part (b), consider the case $u \in V \setminus (V(N_j)\cup V(N_\ell))$. We claim that the entries $uv$ and $vu$ in $\tilde{M}_j$ and $\tilde{M}_\ell$ are 0. Since $u \in V \setminus (V(N_j)\cup V(N_\ell))$, either $\mathcal{T}(u) \cap \mathcal{T}_j\neq \emptyset$ and $u \in \pi_1(D_j)$, or $\mathcal{T}(u) \cap \mathcal{T}_j=\emptyset$ by Lemma~\ref{tech_lemma}(a) and (b). In the former case, the entries $uv$ and $vu$ in $\tilde{M}_j$ are 0 because of Lemma~\ref{tech_lemma2}(c) (applied to $u$). In the latter case, they are 0 because of Lemma~\ref{tech_lemma2}(d). Of course, the same arguments may be applied to $\tilde{M}_\ell$, establishing the claim. At this point, we know that the entries $uv$ and $vu$ in $\hat{M}_i$ are 0. This implies that no row and column operation performed by~\texttt{JoinBox} (which are the operations that turn $\hat{M}_i$ into $\tilde{M}_i$) can modify these entries, as they can only be modified when a multiple of some row or column $w$ is added to row or column $v$, where $w$ has type-i in $N_j$ or $N_\ell$. In other words, for entry $uv$ to become nonzero, entry $uw$ must have become nonzero earlier for another $w \in V_1(N_j) \cup V_1(N_\ell)$, and this never happens because they are all zero in $\hat{M}_i$.

We now consider part (c), and fix $(v,d_v) \in D_i$, where $\mathcal{T}_i \cap \mathcal{T}(v)\neq \emptyset$. First suppose that $v$ was diagonalized during step $i$, that is,  $(v,d_v) \in D_i \setminus D_{i-1}$, then $v$ started as a type-i row in $\hat{M}_i$ and its row and column became 0 after the operations of~\texttt{JoinBox} where performed. By the discussion in the proof of (b), any entries $uv$ and $vu$ of $\hat{M}_i$ such that $u \notin V(N_j) \cup V(N_\ell)$ are equal to 0, and they cannot be modified by the operations that turn $\hat{M}_i$ into $\tilde{M}_i$. If $u \in V(N_j) \cup V(N_\ell)$, the entries $uv$ and $vu$ in the process of obtaining $\tilde{M}_i$ are precisely the entries that we see as we modify $N_i^\ast$, and so they are 0 in $\tilde{M}_i$ because  they became 0 in $N_i^\ast$. Next suppose that $v \in D_{i-1}$. The argument to prove Lemma~\ref{tech_lemma}(d) for introduce nodes can be used to show that either $v \in D_j$ and $\mathcal{T}(v)$ is a subtree of $\mathcal{T}_j$ or that $v \in D_\ell$ and $ \mathcal{T}(v)$ is a subtree of $\mathcal{T}_\ell$. We can easily apply Lemma~\ref{tech_lemma2}(c) to step $j$ or $\ell$, and use the structure of $\hat{M}_i$ described above to conclude that the row and column associated with $v$ is diagonalized with diagonal entry $d_v$ in $\hat{M}_i$. Using the arguments of part (b), this cannot be modified when producing $\tilde{M}_i$.

We conclude with part (d). Let $v$ be such that $\mathcal{T}_i \cap \mathcal{T}(v)=\emptyset$. As $\mathcal{T}_j$ and $\mathcal{T}_\ell$ are contained in $\mathcal{T}_i$, we may apply Lemma~\ref{tech_lemma2}(d) to steps $j$ and $\ell$, and use the structure of $\hat{M}_i$ described above to conclude that the row and column associated with $v$ in $\hat{M_i}$ are 0. Since $v \notin V(N_j) \cup V(N_\ell)$, this cannot be modified when producing $\tilde{M}_i$.

\vspace{10pt}

\noindent \textbf{Case 4 (Forget).} If $i$ has type Forget, let $j$ be its child and call $v^\ast$ the vertex in $B_j \setminus B_i$. Step $i$ starts with a matrix $N_i^\ast$ that is obtained from $N_j$ by updating all entries $uv^\ast$ and $v^\ast u$ such that $u \in B_j$.

In the following, let $M_{v^\ast}$ be the matrix whose rows and columns are labeled by $V$ such that all entries are 0 except for $M_{v^\ast}[u,v^\ast]=M_{v^\ast}[v^\ast,u]=m_{uv^\ast}$, where $u \in B_j$.
\begin{claim}\label{claim:forgetbox}
The matrix $M_i$ in~\eqref{defmi} satisfies $M_i=M_j+M_{v^\ast}$, and the row and column of $v^\ast$ in $M_i$ and $M$ coincide.
\end{claim}

\begin{proof}
Because $B_i \subset B_j$, any vertex $u$ that satisfies $\mathcal{T}(u) \cap \mathcal{T}_i \neq \emptyset$ also satisfies $\mathcal{T}(u) \cap \mathcal{T}_j \neq \emptyset$  By the definition of $M_i$, we have $M_i[u,v]=M_j[u,v]$ and  $M_{v^\ast}[u,v]=0$ except when $\{u,v\} \subseteq B_j$, but $\{u,v\} \not\subseteq B_i$. The latter happens if $\{u,v\} \subseteq B_j$ and $v^\ast \in \{u,v\}$. In this situation, we have $M_i[u,v]=m_{uv}$, $M_j[u,v]=0$ and $M_{v^\ast}[u,v]=m_{uv}$, as required.

For the second claim, fix $u \in V$. First assume that $\{u,v^\ast\} \in B_\ell$ for some node $\ell$. Since $i$ forget $v^\ast$, $\ell$ is a node of $\mathcal{T}_j$ and therefore $\mathcal{T}(u) \cap \mathcal{T}_i \neq \emptyset$, $\mathcal{T}(v^\ast) \cap \mathcal{T}_i \neq \emptyset$ and $\{u,v^\ast\} \not\subseteq B_i$. By definition of $M_i$, $M_i[u,v^\ast]=M_i[v^\ast,u]=m_{uv^\ast}$, as required.

Next assume that $\{u,v^\ast\} \notin B_\ell$ for all $\ell \in V(\mathcal{T})$. By the definition of tree decomposition, $m_{uv^\ast}=m_{v^\ast u}=0$, and therefore   $M_i[u,v^\ast]=M_i[v^\ast,u]=m_{uv^\ast}$ regardless of the case that applies to $u,v^\ast$ in the definition of $M_i$.
\end{proof}

As in previous cases, we prove that Lemmas~\ref{tech_lemma0},~\ref{tech_lemma} and~\ref{tech_lemma2} still hold after step $i$ is performed.

We start with Lemma~\ref{tech_lemma0}. Part (a) again follows from the fact that $\tilde{M}_i$ is obtained from $M_i$ by a sequence of elementary row and column operations, where each row operation is followed by the same column operation. For part (b), note that step $i$ does not remove any elements from the global array, so $D_{i-1} \subseteq D_i$. On the other hand, it may append one or two pairs $(u,d_u)$ to the global array. For any such pair $(u,d_u)$, $u=v^\ast$ or $u \in V_1(N_j)$, so that $u \in V(N_j)$, which is disjoint from $\pi_1(D_j)$ by Lemma~\ref{tech_lemma}(d) applied to step $j$. As seen in the proof of Lemma~\ref{tech_lemma}(d) for introduce nodes, $u$ cannot lie in $\phi(D_{i-1} \setminus D_j)$, and $\pi_1$ is injective over $D_{i}$. To see that part (c) holds, note that a multiple of a row $v$ may only be added to another row in operations of \texttt{ForgetBox} if $v=v^\ast$ or $v \in V_1(N_j)$. This ensures that, at the end of the step, $v \in \pi_1(D_i) \cup V_1(N_i)$.

Next, we consider Lemma~\ref{tech_lemma}. For part (a), let $v$ such that $\mathcal{T}(v) \cap \mathcal{T}_i \neq \emptyset$. Since $B_i \subset B_j$, we have  $\mathcal{T}(v) \cap \mathcal{T}_j \neq \emptyset$. By Lemma~\ref{tech_lemma}(a) applied to step $j$, $v \in V(N_j) \cup \pi_1(D_j)$. We have seen that $D_j \subseteq D_i$ in Lemma~\ref{tech_lemma0}(b). If $v \in V(N_j)$, the description of \texttt{ForgetBox} ensures that $v \in \pi_1(D_i \setminus D_{i-1}) \cup V(N_i)$, and the result holds.  Next we prove (d). Applying Lemma~\ref{tech_lemma}(d) to step $j$, we know that $V_1(N_j)$, $V_2(N_j)$ and $\pi_1(D_j)$ are mutually disjoint. As we argued above, no element in  $\phi(D_{i-1} \setminus D_j)$ can lie in $V(N_i) \subset V(N_j)$. The only element in $V_2(N_j)\setminus V_2(N_i)$ is $v^\ast$, and this element is either in $\pi_1(D_i)$ or in $V_1(N_i)$ at the end of the step. Also, no element $v^\ast$ lies in $V_1(N_i)\setminus V_1(N_j)$ and all elements in $V_1(N_j)\setminus V_1(N_i)$ are added to $\pi_1(D_i)$.  As a consequence, $V_1(N_i)$, $V_2(N_i)$ and $\pi_1(D_i)$ are mutually disjoint, as required for part (d). For part (c), fix $v \in \pi_1(D_{i}) \cup V_1(N_i)$ such that $\mathcal{T}_i \cap \mathcal{T}(v) \neq \emptyset$. By the description of \texttt{ForgetBox} and arguments as above, we conclude that $v \in \pi_1(D_j) \cup V_1(N_j) \cup \{v^\ast\}$. If $v \neq v^\ast$, Lemma~\ref{tech_lemma}(c) applied to step $j$ tells us that $\mathcal{T}(v)$ is a subtree of $\mathcal{T}_j$, and hence of $\mathcal{T}_i$. If $v=v^\ast$, then $\mathcal{T}(v)$ is a subtree of $\mathcal{T}_j$ because $j$ is the child of the node that forgets $i$. This concludes the proof.

To prove Lemma~\ref{tech_lemma2}, we first explain the relationship between the entries of $N_i^\ast$ at the beginning of the step and a matrix that is congruent to $M_i$. The first part of Claim~\ref{claim:forgetbox} tells us that $M_i=M_j+M_{v^\ast}$, so that the only entries that have been modified from $M_j$ to $M_i$ are in positions $uv$ where $u,v \in B_j$. By Lemma~\ref{tech_lemma0}(c), if the row (or column) associated with $w$ has been added to another row (or column) while processing node $\ell$ in $\mathcal{T}_j$, then this row has type-i in $N_\ell$ or has been diagonalized at step $\ell$. By Lemma~\ref{tech_lemma}(b) and (c), we conclude that $\mathcal{T}(w)$ is contained in $\mathcal{T}_j$, and by Lemma~\ref{tech_lemma}(d), $w \notin B_j$. In other words, rows and columns that are used to modify other rows and columns are the same in $M_i$ and $M_j$. This means that performing the same row and column operations to $M_j$ and $M_i$ changes entries in exactly the same ways. Let $\hat{M}_i$ denote the matrix obtained from $M_i$ by performing the elementary row and column operations performed for nodes in $\mathcal{T}_j$. Our discussion ensures that
\begin{equation}\label{eq:forgetbox}
\hat{M}_i=\tilde{M}_j+M_{v^\ast}.
\end{equation}
In particular, we derive the following by Lemma~\ref{tech_lemma2} applied at step $j$:
\begin{itemize}
\item[(I)] If $u,v \in B_j$, then $N_i^\ast[u,v]=\hat{M}_i[u,v]$.

\item[(II)] If $v \in V_1(N_j)$, then the following holds for any $u \in V$. If $u \in B_j$, $N_i^\ast[u,v]=N_i^\ast[v,u]=\hat{M}_i[u,v]=\hat{M}_i[v,u]$. If $u \notin B_j$, then $\hat{M}_i[u,v]=\hat{M}_i[v,u]=0$.

\item[(III)] If $(v,d_v) \in D_j$ and $\mathcal{T}_j \cap \mathcal{T}(v) \neq \emptyset$, then $\hat{M}_i[v,v]=d_v$ and $\hat{M}_i[u,v]=\hat{M}_i[v,u]=0$ for $u \in V \setminus \{v\}$.

\item[(IV)] If $\mathcal{T}_j \cap \mathcal{T}(v) = \emptyset$, then $\hat{M}_i[u,v]=\hat{M}_i[v,u]=0$ for all $u \in V$.
\end{itemize}

We also claim that
\begin{itemize}
\item[(V)] If $u \notin V(N_j)$, then $\hat{M}_i[u,v^\ast]=\hat{M}_i[v^\ast,u]=0$.
\end{itemize}
Indeed, fix $u \notin V(N_j)$. The row and column of $u$ are the same in $M_j$ and $M_i$ by Claim~\ref{claim:forgetbox}. If $\mathcal{T}(u) \cap \mathcal{T}_j =\emptyset$, then the row and column associated with $u$ in $\tilde{M_j}$ is 0 by Lemma~\ref{tech_lemma2}(d). So $\hat{M}_i[u,v^\ast]=\hat{M}_i[v^\ast,u]=0$ by~\eqref{eq:forgetbox}. If $\mathcal{T}(u) \cap \mathcal{T}_j \neq \emptyset$, then $u \in V(N_j) \cup \pi_1(D_j)$ by Lemma~\ref{tech_lemma}(a). Since $u \notin V(N_j)$, we get $u \in \pi_1(D_j)$, so that the row and column associated with $u$ in $\tilde{M_j}$ is 0 except possibly for the diagonal element, which is equal to $\pi_2(\pi^{-1}(u))=d_u$. Again, $\hat{M}_i[u,v^\ast]=\hat{M}_i[v^\ast,u]=0$ by~\eqref{eq:forgetbox}.

We use (I)-(V) to prove Lemma~\ref{tech_lemma2}. We start with (a). For any $u,v \in B_i$, $N_i^\ast$ and $\hat{M}_i$ are the same for entries $u'v'$ such that $u',v' \in V(N_j)$. Since the same row and column operations are performed to produce $N_i$ and $\tilde{M}_i$ from $N_i^\ast$ and $\hat{M}_i$, respectively, we get $N_i[u,v]=\tilde{M}_i[u,v]$.

To prove Lemma~\ref{tech_lemma2}(b), fix $v \in V_1(N_i)$. If $u \in V(N_i)$, then $N_i[u,v]=N_i[v,u]=\tilde{M}_i[u,v]=\tilde{M}_i[v,u]$ for the same reason of part (a), namely that the same row and column operations are performed to $N_i^\ast$ and $\hat{M}_i$. Next assume that $u \notin V(N_i)$. The cases $u \in V(N_j)$ and $u \notin V(N_j)$ will be considered separately. We start with the latter. The entries $uv$ and $vu$ of any matrix produced while getting from $\hat{M}_i$ to $\tilde{M}_i$ can only be modified by adding a row or column to row or column $v$. During \texttt{ForgetBox}, this could be an operation as in~\eqref{eq:1c} (in the case $v=v^\ast$) or as in~\eqref{eq_case21} (in the case $v \neq v^\ast$). In both cases, the row or column being added is associated with some $u' \in V_1(N_j)$, and the entries $uu'$ and $u'u$ are 0 by item (II) above\footnote{Of course, (II) refers only to matrix $\hat{M}_i$ and not to intermediate matrices in the process of getting from $\hat{M}_i$ to $\tilde{M}_i$. However, it is easy to prove that entries of this type are always equal to 0 using an argument as in the proof of Lemma~\ref{tech_lemma2}(b) for \texttt{JoinBox}.}. This argument implies that row and column operations of  \texttt{ForgetBox} cannot modify entries of $\hat{M}_i$ that lie ``outside the box'' $N_j$, that is, entries $u'v'$ such that $u'\notin V(N_j)$ or $v'\notin V(N_j)$.

To conclude part (b), assume that $u \in V(N_j) \setminus V(N_i)$. This can happen if $u \neq v^\ast$, we are in Case 2 and $u$ is the vertex associated with the rightmost entry in $x_{v^\ast}$ or if $u=v^\ast$ and we are in any case other than Subcase 1(c), when $v$ is inserted in $N_i^{(1)}$. In both cases, the conclusion will follow from our proof of part (c), as we shall prove that the row and column associated with $u$ has been diagonalized at step $i$, so that the entries $uv$ and $vu$ mentioned here must be 0.

 For part (c), fix $(v,d_v) \in D_i$ such that $\mathcal{T}_i \cap \mathcal{T}(v) \neq \emptyset$. By our argument to prove Lemma~\ref{tech_lemma2}(c) for introduce nodes, we know that $(v,d_v) \notin D_{i-1}\setminus D_j$, so that row $v$ has either been diagonalized at step $i$ or it had already been diagonalized at step $j$. First suppose that $(v,d_v)\in D_i \setminus D_{i-1}$, which means that either $v=u$, the vertex associated with the rightmost entry in $x_{v^\ast}$ or $v=v^\ast$ and we are in any case other than Subcase 1(c), when $v$ is inserted in $N_i^{(1)}$. We start with $v=u$. The operations performed for case 2 make sure that $u$ is diagonalized within $N_i^\ast$, and therefore the entry $uu$ is given by $d_u$ and the other entries become 0.  The relationship between $N_i^\ast$ and $\hat{M}_i$ ensures that $\hat{M}_i[u,u]=d_u$ and that $\hat{M}_i[u,w]=\hat{M}_i[w,u]=0$ for $w \in V(N_j) \setminus \{u\}$. If $w \in V\setminus V(N_j)$ then the entries $uw$ and $wu$ are 0 in $\hat{M}_i$. As argued in part (c), these entries cannot be modified when a row or column associated with a vertex in $V_1(N_j)$ is added to row $u$. Moreover, by (V) above, the entries $v^\ast w$ and $wv^ast$ in $\hat{M}_i$ are 0, so that the entries $uw$ and $wu$ cannot be modified if a multiple of row or column $v^\ast$ is added to row $u$. A similar argument deals with the case $v=v^\ast$. To finish the proof of part (c), suppose that $(v,d_v) \in D_j$. This means that $v \notin V(N_j)$ and that the row and column associated with $v$ in $\tilde{M}_j$, and therefore in $\hat{M}_i$, is 0, except possibly for the entry $\hat{M}_i[v,v]=d_v$. No row or column is added to row or column $v$ during \texttt{ForgetBox}, and an entry $wv$ or $vw$ cannot be modified by the addition of a row or column to the row or column associated with $w$ by the arguments of the previous case.

For part (d), let $v$ be such that $\mathcal{T}_i \cap \mathcal{T}(v)=\emptyset$, so that $\mathcal{T}_j \cap \mathcal{T}(v)=\emptyset$. By (IV) and the arguments above, we easily conclude that the row and column associated with $v$ in $\hat{M}_i$ are 0, and that they cannot be modified during row and column operations performed by \texttt{ForgetBox}, as these entries are outside the box.

This concludes the induction and shows that Lemma~\ref{tech_lemma0}, Lemma~\ref{tech_lemma} and Lemma~\ref{tech_lemma2} hold throughout the algorithm. As argued above, this shows that \texttt{CongruentDiagonal} produces the correct output.

\section{Running-time of the algorithm} \label{sec:analysis}

To finish the paper, we show that the running-time of \texttt{CongruentDiagonal} given in Theorem~\ref{main_theorem} is correct. If the algorithm is given a tree decomposition $\mathcal{T}$ with $m$ nodes and width $k$ of the underlying graph $G$ of matrix $M$, it first computes a nice tree decomposition $\mathcal{T}'$ of the same width $k$ with fewer than $4n-1$ nodes in time $O(k(m+n))$ by Lemma~\ref{lemma_nice:tree}. The proof of Lemma~\ref{lemma_nice:tree} implies that the number of nodes of each type in $\mathcal{T}'$ is at most $n$.

We consider the number of operations performed at each type of node. \texttt{LeafBox} initializes a matrix in $O(k^2)$ trivial steps. \texttt{IntroduceBox} uses $O(k)$ steps to create a row and column filled with zeros.

For \texttt{JoinBox} and \texttt{ForgetBox}, the main cost comes from row and column operations. As the matrices involved have dimension at most $k+1$ each such row or column operation requires $O(k)$ sums and multiplications. Regarding \texttt{ForgetBox}, when $v$ is forgotten, either $v$ is turned into a type-i vertex, or its row and column, and possibly the row and column of another vertex $u$, are diagonalized. The latter requires $O(k)$ row and column operations. If $v$ is turned into a type-i vertex, then inserting it into the matrix $N_i^{(1)}$ in row echelon form takes at most $k+1$ row operations. \texttt{JoinBox} can be most time consuming. To insert just one row vector into a matrix in row echelon form with $k+1$ rows, and preserving this property by adding multiples of one vector to another, can require up to $k+1$ row operations. Recall that each such operation can be done in time $O(k)$. To combine two matrices in row echelon form into one such matrix, up to $k+1$ row vectors are inserted. Thus the total time for this operation is $O(k^3)$. This would immediately result in an upper bound of $O(km+k^3 n)$ for the whole computation, where the first term can be omitted if we assume that we already start with a tree decomposition with $O(kn)$ nodes for which the vertices in each bag are sorted.

To realize that $O(k^3 n)$ may be replaced by $O(k^2 n)$, we have to employ a different accounting scheme for the time spent to keep $N_i^{(1)}$ in row echelon form in all applications of \texttt{JoinBox} and \texttt{ForgetBox}. This requires the input graph $G$ to be labeled in a convenient way in terms of the nice tree decomposition $\mathcal{T}$, as will be described below. We first recall some facts about the algorithm:
\begin{itemize}
\item[(a)] The columns of $N_i^{(1)}$ are labeled by type-ii vertices, namely elements of $B_i$, which are always in increasing order. 

\item[(b)] If a new type-ii vertex appears (which may only happen at a node $i$ of type Introduce), the column corresponding to it in $N_i^{(1)}$ is a zero column, meaning that it creates no pivots. 

\item[(c)] If a type-ii row $v$ is forgotten at step $i$ (and $j$ is the child of $i$), a column is removed from $N_j^{(1)}$ in the process of defining $N_i^{(1)}$. By the description of \texttt{ForgetBox} if a type-i row $u$ had a pivot in $N_j^{(1)}$ at the column indexed by $v$, then both $u$ and $v$ are diagonalized at step $i$ and no other pivots change in this step.
\end{itemize}

For nodes of type Forget, we count the row operations performed to insert a new row into $N_i^{(1)}$ and to preserve it in row echelon form in terms of the vertex $v$ being forgotten. As mentioned above, inserting $v$ takes $O(k)$ row and column operations, each requiring $O(k)$ field operations. Overall, the number of such operations is $O(k^2n)$. 

For nodes of type Join, we count the row operations in terms of the vertex that indexes the row $v$ being modified by the operation. This operation creates at least one 0 entry in row $v$, meaning that either the type-ii label of its pivot increases or the row becomes a zero row. Of course, row $v$  becomes a zero row at most once, as it becomes diagonalized when this happens. We claim that, if the vertices of the input graph are conveniently labeled, the label of the pivot of row $v$ cannot increase more than $k$ times. If this holds, a vertex $v$ can be involved in at most $O(k)$ row operations in nodes of type Join for the entire algorithm, each requiring $O(k)$ field operations. Overall, the number of such operations is $O(k^2n)$, giving the desired result. 

To conclude the proof, we establish that the vertex set of $G$ may indeed be labeled so that the label of the pivot of any type-i row $v$ cannot increase more than $k$ times. Given the input graph $G$ and an input nice tree decomposition $\mathcal{T}$ with $O(n)$ nodes, we first modify $\mathcal{T}$ to obtain a nice tree decomposition with the additional properties below, which may be done in time $O(kn)$ with the procedure described in our proof of Lemma~\ref{lemma_nice:tree} (see appendix):
\begin{itemize}
\item[(i)] If $i$ and $j$ are both nodes of type Join and $i$ is an ancestor of $j$, then $|B_i| \leq |B_j|$.

\item[(ii)] If $j$ is a node of type Join and $i$ is the first node of type Join on the path between $j$ and the root, then the path between $j$ and $i$ is either a single edge, a path for which the internal nodes have type Forget, or a path whose internal nodes first alternate between types Forget and Introduce (starting with Forget), and are then possibly followed by a sequence of nodes of type Forget.
\end{itemize}  

Next, given the bottom-up labeling of $\mathcal{T}$, we relabel the vertex set of $G$ so that the first vertex that is forgotten (by a node of type Forget) is labeled $n$, the second is labeled $n-1$, and so on. Let $i$ be a node of type Forget and $j$ be its child. Observe that, with this labeling, the last column of $N_j^{(1)}$ is deleted in the course of producing $N_i^{(1)}$.

Assuming that the vertex set of $G$ is labeled in this way, we now consider a vertex $v$ that became a type-i row at step $i$. For any node $\ell$ for which $v$ has type-i, we say that $v$ has \emph{distance $d$ to the end of the box $N_\ell$} if the row of $v$ in $N_\ell^{(1)}$ has $d$ entries after its pivot. Since the width of the decomposition is $k$, $N_\ell^{(1)}$ has at most $k+1$ columns, so that the distance of $v$ to the end of any box is always bounded above by $k$. This is true for the first node $\ell_1$ of type Join such that $v$ has type-i. Let $d_0$ be the distance from $v$ to the end of the box at the start of step $\ell_1$ (see~\eqref{eq:formofNast}). During this step, each time $v$ is modified by a row operation, the label of its pivot increases by at least one. In particular, $v$ is involved in at most $d_0 \leq k$ operations if its row becomes a zero row. Suppose that it does not become a zero row and let $d_1$ be the distance from $v$ to the end of the box $N_{\ell_1}^{(1)}$. Clearly, $1\leq d_1 \leq d_0-a_1$, where $a_1$ is the number of row operations that have been performed upon $v$ during step $\ell_1$. If $v$ does not appear in another node of type Join, we are done because $a_1 \leq d_0-1< k$. So assume that it appears in another node of type Join, let $\ell_2$ be the first such node with respect to the bottom-up ordering. Along the path between $\ell_1$ and $\ell_2$, the distance from $v$ to the end of the corresponding box \emph{may increase} by one for each node of type Introduce (because it creates a new column that may appear after the pivot of row $i$), and \emph{always decreases} by one for each node of type Forget (because the vertex labeling ensures that the column corresponding to the forgotten vertex is always the last row of the box). Since there are at least as many nodes of type Forget as nodes of type Introduce along the path between $\ell_1$ and $\ell_2$, when step $\ell_2$ starts, the distance from $v$ to the end of the box is at most $d_1$. As before, if the row of $v$ becomes a zero row during this step, it may be involved in at most $d_1$ row operations, and the overall number of operations involving $v$ in nodes of type Join is at most $d_1+a_1 \leq d_0 \leq k$. Otherwise, its distance $d_2$ to the end of the box $N_{\ell_2}$ satisfies $1 \leq d_2 \leq d_1-a_2 \leq d_0-(a_1+a_2)$, where $a_2$ is the number of row operations that have been performed upon $v$ during this step. This implies that $a_1+a_2 \leq d_0-1$. An inductive application of this argument implies that the claim holds.

To conclude the paper, we note that the algorithm works for graphs with arbitrary vertex labelings associated with arbitrary nice tree decompositions whose roots have an empty bag. Using a special ordering and a special nice tree decomposition is only needed to bound the number of operations. However, even though we can craft examples where the pivot of a single row can move more than $k$ times, we are not aware of an example where more than $O(k^2 n)$ field operation are performed overall for nodes of type \texttt{JoinBox}.  In fact, in most practical applications, the set of type-i rows tends to remain quite small throughout the algorithm.   

\bibliographystyle{abbrv}
\bibliography{paper_treewidth}

\appendix

\section{Proof of Lemma~\ref{lemma_nice:tree}}

For completeness, we provide a proof of Lemma~\ref{lemma_nice:tree}.

Let $\mathcal{T}$ be a tree decomposition of $G$ with width $k$ with $m$ nodes, and fix an arbitrary node $i$ as its root. At the beginning, we sort the vertices in each bag in increasing order. This could trivially be done in time $O((k \log k)m)$ for any pre-determined order, but with our assumption that $V=\{1,\dots,n\}$, this can be done in time $O(km)$ by first producing a list of tree nodes in whose bags a vertex $v$ appears, simultaneously for every vertex $v$. We modify the tree decomposition in a sequence of depth-first traversals. In the first traversal, every node whose bag is contained in the bag of its parent is merged with the parent, that is, the node is removed and its children are connected to the parent.
\begin{figure}[t]
\begin{tikzpicture}
  [scale=1,auto=left,every node/.style={circle,scale=1}]
  \node[draw,circle,fill=blue!25, inner sep=1.5] (a) at (-3.5,0) {$13$};
  \node[draw,circle,fill=blue!25, inner sep=1.5] (b) at (-2,1) {$24$};
    \node[draw,circle,fill=blue!10, inner sep=1.5] (d) at (-2,-1) {$12$};
    \node[draw,circle,fill=blue!10, inner sep=1.5] (c) at (-.5,1.5) {$34$};
    \node[draw,circle,fill=blue!10, inner sep=1.5] (cc) at (-.5,.5) {$14$};

  \node[draw,circle,fill=blue!25, inner sep=1.5] (e) at (2.5,0) {$13$};
  \node[draw,circle,fill=blue!25, inner sep=1.5] (f) at (4,-1) {$13$};
   \node[draw,circle,fill=blue!25, inner sep=1.5] (g) at (4,1) {$13$};
   \node[draw,circle,fill=blue!10, inner sep=1.5] (j) at (5.5,-1) {$12$};

   \node[draw,circle,fill=blue!25, inner sep=1.5] (h) at (5.5,1) {$24$};
   \node[draw,circle,fill=blue!25, inner sep=1.5] (i) at (7,2) {$24$};
   \node[draw,circle,fill=blue!25, inner sep=1.5] (k) at (7,0) {$24$};
   \node[draw,circle,fill=blue!10, inner sep=1.5] (l) at (8.5,2) {$34$};
   \node[draw,circle,fill=blue!10, inner sep=1.5] (ll) at (8.5,0) {$14$};
   \node (m) at (.5,.5) {$\longrightarrow$};
   \node (n) at (-4.25,0) {};
   \node (o) at (1.75,0) {};
  \path (a) edge node[below]{} (b)
          (b) edge node[below]{} (c)
          (b) edge node[below]{} (cc)
        (a) edge node[below]{} (d)
        (e) edge node[below]{} (f)
        (e) edge node[below]{} (g)
        (f) edge node[below]{} (j)
        (h) edge node[below]{} (i)
        (h) edge node[below]{} (k)
        (g) edge node[below]{} (h)
        (l) edge node[below]{} (i)
        (h) edge node[below]{} (i)
        (ll) edge node[below]{} (k)
        (a) edge node[below]{} (n)
        (o) edge node[below]{} (e);
        ;
\end{tikzpicture}
\caption{Third traversal, where nodes are replaced by a binary tree. The root of the decomposition is located towards the left.}
     \label{fig-tr3}
\end{figure}
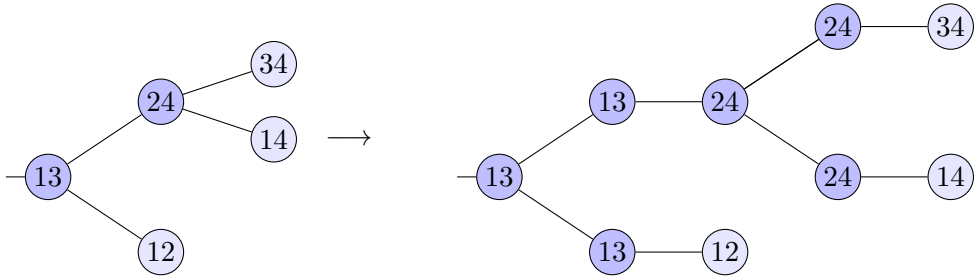

In the second traversal, whenever the bag of a node has size less than the size of the bag of its parent, we add some vertices from the parent's bag until both bags have the same size. At this point, the bags of children are at least as large as the bags of their parents, and they are never contained in the bag of their parent. In the third traversal, each node $i$ with $c \geq 2$ children and bag $B_i$ is replaced by a binary tree with exactly $c$ leaves whose nodes are all assigned the bag $B_i$. Each child of $i$ in the original tree becomes the single child of one of the leaves of this binary tree. At this point, all nodes have at most two children and those with two children are Join nodes. Figure~\ref{fig-tr3} illustrates this transformation. In the fourth traversal, for any node $i$ with a single child $j$, if necessary, replace the edge $\{i,j\}$ by a path such that the bags of the nodes along the path differ by exactly one vertex in each step. This is done from $j$ to $i$ by a sequence of nodes, starting with a Forget node, then alternating between Introduce and Forget nodes, and possibly ending with a sequence of Forget nodes in case the child's bag is larger than its parent's. To ensure that property (3) of a tree decomposition is satisfied, each vertex in the symmetric difference of the original bags $B_i$ and $B_j$ produces a single Forget or Introduce node. Figure~\ref{fig-tr4} illustrates this transformation.

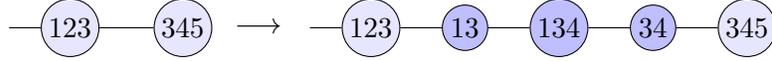
\begin{figure}[h!]
\begin{tikzpicture}
  [scale=1,auto=left,every node/.style={circle,scale=1}]
  \node[draw,circle,fill=blue!10, inner sep=1.5] (a) at (0,2) {$123$};
  \node[draw,circle,fill=blue!10, inner sep=1.5] (b) at (1.5,2) {$345$};
    \node[draw,circle,fill=blue!10, inner sep=1.5] (d) at (4,2) {$123$};
  \node[draw,circle,fill=blue!25, inner sep=1.5] (e) at (5.25,2) {$13$};
  \node[draw,circle,fill=blue!25, inner sep=1.5] (f) at (6.5,2) {$134$};
   \node[draw,circle,fill=blue!25, inner sep=1.5] (g) at (7.75,2) {$34$};
   \node[draw,circle,fill=blue!10, inner sep=1.5] (h) at (9,2) {$345$};
    \node (m) at (2.5,2) {$\longrightarrow$};
     \node (n) at (-1,2) {};
   \node (o) at (3,2) {};
  \path (a) edge node[below]{} (n)
   (o) edge node[below]{} (d)
  (a) edge node[below]{} (b)
        (d) edge node[below]{} (e)
        (e) edge node[below]{} (f)
        (f) edge node[below]{} (g)
        (g) edge node[below]{} (h);
\end{tikzpicture}
\caption{Fourth traversal, which creates an alternating path of Forget and Introduce nodes. The root of the decompositions is located towards the left.}
     \label{fig-tr4}
\end{figure} 
To finish our construction, if the root $i$ has a nonempty bag $B_i$, we append a path to the root where each node is a forget node, until we get an empty bag, which becomes the new root. Note that at most $k+1$ nodes are appended to the path, as $|B_i|\leq k+1$. We call this nice tree decomposition $\mathcal{T}'$.  We claim that $\mathcal{T}'$ is a nice tree decomposition of $G$ with the same width as $\mathcal{T}$. It is easy to see that properties (1) and (3) in the definition of tree decomposition are not violated after each traversal, and that all nodes in $\mathcal{T}$ have one of the types in the definition of nice tree decomposition. Moreover, bag sizes have only been increased in the second traversal, but the size is always bounded by the size of a bag that was already in the tree. Finally, we see that $\mathcal{T}'$ is a decomposition of $G$, in the sense that property (2) in the definition of tree decomposition is satisfied. To this end, if two vertices $u,v \in V(G)$ lie in a bag of $\mathcal{T}$, note that one of the bags of the new tree must contain the largest bag originally containing $u$ and $v$, and therefore it is a tree decomposition of $G$.

Next, we count the number $m'$ of nodes of $\mathcal{T}'$. The Leaf nodes have degree 1, the Join nodes have degree 3 (unless one of the Join nodes is the root and has degree 2) and Forget and Introduce nodes have degree 2 (unless one of the Forget nodes is the root and has degree 1). Let $m'_F$, $m'_I$, $m'_J$ and $m'_L$  denote the number of Forget, Introduce, Join and Leaf nodes in $\mathcal{T}'$, respectively. Since the sum of degrees is $2m'-2$, we have
$$2(m'_F+m'_I+m'_J+m'_L)-2=3m'_J+2(m'_I+m'_F)+m'_L-1,$$
which leads to $m'_J= m'_L-1$. Recall that every vertex is forgotten exactly once, so $m'_F=n$. To see that $m'_L \leq n$, first note that $m'_L=1$ if there are no join nodes. Otherwise, the first and the fourth traversals ensure that there is at least one Forget node on the path between each Leaf node and the first Join node on its path to the root, so that  $m'_L\leq m'_F\leq n$. It follows that $m'_J\leq n-1$. The fourth traversal ensures that the single child of every Introduce node is a Forget node, so that we also have $m'_I\leq m'_F \leq n$. Therefore $\mathcal{T}'$ has at most $4n-1$ nodes.

To conclude the proof, we compute the time required for this transformation. In the first traversal, for each node of $\mathcal{T}$, we need to check whether its bag is contained in its parent's. This is done in time $O(km)$, since we are assuming that the bags are sorted. Note that the other traversals cannot decrease the number of nodes in the tree, so the trees traversed in the remaining steps have size $O(n)$ by the above discussion. The second traversal also requires comparing the bag of each node with its parent's, so it takes time $O(kn)$. The third traversal requires to verify the degree of each node in the tree and possibly create a number of new nodes with bags of size at most $k+1$. This takes time $O(kn)$ because we know that the number of nodes is at most $4n$ \emph{after} all the new nodes have been created. The final traversal also takes order $O(kn)$ because we only need to verify node degrees, compute the symmetric difference of the content of bags of size at most $k+1$ and create new nodes with bags of size at most $k+1$. We again observe that the number of nodes is at most $4n$ \emph{after} all the new nodes have been created. Overall, the changes performed starting from the first traversal take time $O(km+kn)$.

\end{document}